\documentclass{imamat} 
\usepackage{latexsym, amssymb, enumerate}
\usepackage{hyperref, color,dcolumn}
\usepackage{modamsthm,modnatbib}
\usepackage{lastpage,graphicx,amsfonts}
\usepackage{times,mathptmx,bm,amsmath}

\renewcommand{\theequation}{\arabic{section}.\arabic{equation}}

\renewenvironment{proof}[1][Proof]{\noindent\textit{#1. } }{\hfill$\square$}

 \newtheoremstyle{theorem}{6pt}{6pt}{\rm}{}{\sffamily}{ }{ }{}
 \theoremstyle{theorem}
\newtheorem{theorem}{\sc Theorem}[section]

  \newtheoremstyle{thm}{6pt}{6pt}{\rm}{}{\sffamily}{ }{ }{}
 \theoremstyle{thm}
\newtheorem{thm}{\sc Theorem}[section]

 \newtheoremstyle{lemma}{6pt}{6pt}{\rm}{}{\sffamily}{ }{ }{}
 \theoremstyle{lemma}
 
 \newtheoremstyle{lem}{6pt}{6pt}{\rm}{}{\sffamily}{ }{ }{}
 \theoremstyle{lem}
\newtheorem{lem}{\sc Lemma}[section]

\newtheoremstyle{case}{6pt}{6pt}{\rm}{}{}{. }{ }{}
 \theoremstyle{case}

 \newtheoremstyle{statement}{6pt}{6pt}{\rm}{}{\sffamily}{ }{ }{}
\theoremstyle{statement}

 \newtheoremstyle{corollary}{6pt}{6pt}{\rm}{}{\sffamily}{ }{ }{}
 \theoremstyle{corollary}

  \newtheoremstyle{defi}{6pt}{6pt}{\rm}{}{\sffamily}{ }{ }{}
 \theoremstyle{defi}

  \newtheoremstyle{cor}{6pt}{6pt}{\rm}{}{\sffamily}{ }{ }{}
 \theoremstyle{cor}

\newtheoremstyle{example}{6pt}{6pt}{\rm}{}{\sffamily}{ }{ }{}
\theoremstyle{example}

\newtheoremstyle{remark}{6pt}{6pt}{\rm}{}{\sffamily}{ }{ }{}
\theoremstyle{remark}
\newtheorem{remark}{\sc Remark}[section]

\newtheoremstyle{approximation}{6pt}{6pt}{\rm}{}{\sffamily}{ }{ }{}
\theoremstyle{approximation}

\newtheoremstyle{scheme}{6pt}{6pt}{\rm}{}{\sffamily}{ }{ }{}
\theoremstyle{scheme}

\newtheoremstyle{Algorithm}{6pt}{6pt}{\rm}{}{\sffamily}{ }{ }{}
\theoremstyle{Algorithm}

 \newtheoremstyle{Remark}{6pt}{6pt}{\rm}{}{\sffamily}{ }{ }{}
 \theoremstyle{Remark}

\newtheoremstyle{Lemma}{6pt}{6pt}{\rm}{}{\sffamily}{ }{ }{}
\theoremstyle{Lemma}

\newtheoremstyle{Assumption}{6pt}{6pt}{\rm}{}{\sffamily}{ }{ }{}
\theoremstyle{Assumption}

\newtheoremstyle{Proposition}{6pt}{6pt}{\rm}{}{\sffamily}{ }{ }{}
\theoremstyle{Proposition}

\newtheoremstyle{prop}{6pt}{6pt}{\rm}{}{\sffamily}{ }{ }{}
\theoremstyle{prop}

\newtheoremstyle{rem}{6pt}{6pt}{\rm}{}{\sffamily}{ }{ }{}
 \theoremstyle{rem}
 \newtheorem{rem}{\sc Remark}[section]

\newtheoremstyle{hypo}{6pt}{6pt}{\rm}{}{\sffamily}{ }{ }{}
 \theoremstyle{hypo}

  \newtheoremstyle{Step}{6pt}{6pt}{\rm}{}{}{ }{ }{}
 \theoremstyle{Step}

 \newtheoremstyle{lema}{6pt}{6pt}{\rm}{}{\sffamily}{ }{ }{}
 \theoremstyle{lema}

\newcommand{\reff}[1]{(\ref{#1})}
\newcommand{\ol}[1]{\overline{#1}}

\newcommand{\Av}{\tilde{\rm A}_v}
\newcommand{\Aw}{\tilde{\rm A}_w}
\newcommand{\PB}[1]{{\rm P}_N{\rm B}\brkts{#1}}
\newcommand{\QB}[1]{{\rm Q}_N{\rm B}\brkts{#1}}

\newcommand{\g}{\mathrm{g}}

\newcommand{\cg}[1]{\mathcal{#1}}

\newcommand{\av}[1]{\left|#1\right|}

\newcommand{\Ll}[1]{\left\|#1\right\|}
\newcommand{\N}[2]{\left\|#1\right\|_{#2}}

\newcommand{\brkts}[1]{\left(#1\right)}
\newcommand{\ebrkts}[1]{\left[#1\right]}

\newcommand{\brcs}[1]{\left\{#1\right\}}

\newcommand{\bsplitl}[2]{
\begin{equation}
\begin{split}
#1
\end{split}
\label{#2}
\end{equation}}
\newcommand{\bsplit}[1]{
\begin{equation*}
\begin{split}
#1
\end{split}
\end{equation*}}

\newcommand{\new}[1]{#1} 

\DeclareGraphicsExtensions{.png,.eps,.ps}
\graphicspath{{archive_pics/}}
\renewcommand{\theequation}{\thesection.\arabic{equation}}
\numberwithin{equation}{section}

\begin{document}

\title{A new mode reduction strategy for the generalized Kuramoto-Sivashinsky equation}
\author{{\sc M. Schmuck$^{1,2}$, M. Pradas$^1$, G. A. Pavliotis$^2$, and S.~Kalliadasis$^1$}\\[2pt]
$^1$ Department of Chemical Engineering, Imperial College, London SW7 2AZ, UK\\[2pt]
$^2$ Department of Mathematics, Imperial College, London SW7 2AZ, UK\\[2pt]
}
\markboth{M.~Schmuck et al.}{A new mode reduction method for non-Hamiltonian PDEs}
\maketitle

\begin{abstract}
{Consider the generalized Kuramoto-Sivashinsky (gKS) equation. It is
a model prototype for a wide variety of physical systems, from
flame-front propagation, and more general front propagation in
reaction-diffusion systems, to interface motion of viscous film
flows. Our aim is to develop a systematic and rigorous
low-dimensional representation of the gKS equation. For this
purpose, we approximate it by a renormalization group (RG) equation which
is qualitatively characterized by rigorous error bounds. This
formulation allows for a new stochastic mode reduction guaranteeing
optimality in the sense of maximal information entropy. Herewith,
noise is systematically added to the reduced gKS equation and gives
a rigorous and analytical explanation for its origin.

These new results would allow to reliably perform low-dimensional numerical
computations by accounting for the neglected degrees of freedom in a
systematic way. Moreover, the presented reduction strategy might also be
useful in other applications where classical mode reduction approaches fail
\new{or are too complicated to be implemented}.} {Generalized Kuramoto-Sivashinsky equation,
renormalization group method, stochastic mode reduction}
\end{abstract}


\section{Introduction}
\label{sec:Intr}

We consider abstract evolution equations of the form,
\bsplitl{
\frac{\partial u}{\partial t}
    + {\rm B}(u,u)
    +{\rm A}u
    = 0\,,&
\\u(x,0)
    = u_0(x)\,,
}{absKS} where ${\rm A}$ denotes a general linear operator and ${\rm
B}$ represents a nonlinear term of Burgers' type, i.e., $\new{{\rm B}(u)={\rm B}(u,u)=}uu_x$. Well known
equations in this class include, e.g. the viscous Burgers equation,
the Korteweg-de Vries equation, and the Benney-Lin equation. We
start by performing a formal renormalization group (RG) approach for
the general form in~\reff{absKS}. We subsequently focus on a
rigorous low-dimensional reduction of the generalized
Kuramoto-Sivashinsky (gKS) equation: \bsplitl{
\partial_t u
    +\lambda uu_x
    +\kappa u_{xx}
    +\delta u_{xxx}
    +\nu u_{xxxx}
    = 0\,,
    & \qquad\textrm{in }{\cal P}_\alpha\times ]0,T[\,,
\\
u(x,0)
    = g(x)\,,
    &\qquad\textrm{in }{\cal P}_\alpha\,,
}{KS} where ${\cal P}_\alpha:=]-\alpha\pi,\alpha\pi[$ is a periodic
domain with $\alpha:=\frac{L}{2\pi}$ for an arbitrary period $L>0$
while the solution $u(x,t)\, :\, {\cal P}_\alpha\times
]0,T[\to\mathbb{R}$ of \reff{KS} represents for example
\new{ the fluctuations around a fixed mean height of
a one-dimensional surface above a substrate point $x$ at time $t$ as e.g.~in
a thin film flowing down a vertical wall,
e.g.~\cite{Kalliadasis2011,Marc_JFM,Marc_PoF}.} We also take $g(x)\in
H^q({\cg P}_\alpha)$ for $q\geq 4$, a periodic initial condition, i.e.,
$$g(x+L)=g(x)\,.$$

The gKS equation is of the type~\reff{absKS} with, \bsplitl{ {\rm
B}(u) :=
    {\rm B}(u,u)
    := \lambda uu_x\,,
\qquad\textrm{and}\qquad {\rm A}
    :=\brkts{
    \kappa \partial^2_x
    +\delta\partial^3_x
    +\nu\partial^4_x
    }\,.
}{absOp} A rigorous dimensional reduction of the gKS equation is of
special interest because it is not a Hamiltonian system and does not
have an intrinsic invariant measure. This makes direct application
of stochastic mode reduction strategies difficult, see \cite{Stinis}
for instance.

It is noteworthy that the gKS equation retains the fundamental elements of
any nonlinear process that involves wave evolution: the simplest possible
nonlinearity $u u_x$, instability and energy production $u_{xx}$, stability
and energy dissipation $u_{xxxx}$ and dispersion $u_{xxx}$. We notice that
the nonlinearity arises effectively from the nonlinear correction to the
phase speed, a nonlinear kinematic effect that captures how larger waves move
faster than smaller ones. In the context of thin-film
flows~\cite{Tseluiko2009,Tseluiko2010,Tseluiko2010b,Tseluiko2012}, the terms
$u u_x$, $u_{xx}$, $u_{xxx}$ and $u_{xxxx}$ are due to the interfacial
kinematics associated with mean flow, inertia, viscosity and surface tension,
respectively, with the corresponding parameters, $\lambda$, $\kappa$,
$\delta$ and $\nu$ all positive and measuring the relative importance of
these effects. The ``strength of the nonlinearity", $\lambda$, in particular,
is associated with the scaling for the velocity (and hence time). In
addition, $\int_{{\cal P}_\alpha} u\,dx=u_0$, a measure of the volume of the
liquid, a conservation property for systems whose spatial average does not
drift. A simplified form of \reff{KS} is obtained by appropriately rescaling
$u$, $x$ and $t$ which is equivalent to setting $\lambda=\kappa=\nu=1$ and
keeping the same notation for dimensionless
quantities~\cite{Tseluiko2009,Tseluiko2010,Tseluiko2010b,Tseluiko2012}.

\medskip

As with many nonlinear time-dependent problems in science and
engineering, equations of the form \reff{absKS} are too complex to
be fully resolved and the influence of neglected degrees of freedom
is not clear a priori. This problem exists independently of spatial
dimensions for \reff{absKS} and hence for the gKS equation also.
The reliable resolution of high dimensional problems is a well-known
issue in computational science where one can numerically only deal with a finite number of
degrees of freedom.

Hence, there is a strong need for \new{(finite dimensional/)} dimensionally reduced formulations, which
in turn would allow for studies of long time behavior of physical systems.
Modeling of the ocean-atmosphere, which mainly generates our weather, is one
important example: One has a characteristic timescale of several years for
the ocean in contrast to a couple of days governing atmospheric structures
such as cyclones. As a consequence, a characteristic feature of many physical
systems is the presence of fast and slow degrees of freedom. The relevant
information of a system's long time behavior is often primarily contained in
the slow modes. For Hamiltonian systems, such a mode reduced
mathematical formulation is generally obtained by the Mori-Zwanzig or optimal
prediction techniques as described later on. Here we focus on nonlinear
equations not showing a Hamiltonian-like structure as exemplified by equation
\reff{KS} and we provide a systematic (e.g. via our RG method and maximum
entropy principle) and rigorous (e.g. via error estimates) framework for the
reliable derivation of low-dimensional (/slow-mode) representations of such
equations.

First, we recall the general, often ad hoc, approximation
of decomposing the problem of interest into fast $w$ and slow $v$
modes. For equation \reff{absKS} such a purely formal splitting,
i.e., $u \approx u^\epsilon = v^\epsilon + w^\epsilon$, reads in
standard notation applied in the literature as, \bsplitl{
\frac{\partial}{\partial t}v^\epsilon
    & = f(v^\epsilon,w^\epsilon)\,,
\\
\frac{\partial}{\partial t}w^\epsilon
    & = \frac{1}{\epsilon} g(v^\epsilon,w^\epsilon)\,,
}{SloFas} where the small parameter $0<\epsilon \ll 1$ mediates the timescale
separation. Mode reduction strategies, such as ``adiabatic elimination''
\cite{VanKampen1985}, invariant manifolds \cite{Foias1988}, and optimal
prediction \cite{Chorin1998}, are tools to eliminate the fast modes and
derive ``appropriate'' equations for the slow modes only. We remark that
especially for systems with spatio-temporal chaos (like the gKS equation
\cite{CHOW1995,Tseluiko2010}) such a reduction needs to be carefully
performed in order to not lose the relevant dynamical characteristics of the
full system~\cite{Schmuck2013}. \new{Also the study in~\cite{Marion1989} emphasizes the
importance of careful finite dimensional approximations with computational
schemes by exploiting the structure of Galerkin methods.} The strategy of
defining an invariant manifold is almost classical by now. For example, in
\cite{Foias1988a} the existence of an inertial manifold for the KS equation
(obtained from (\ref{KS}) with $\delta = 0$) is shown. An inertial manifold
is a finite-dimensional, exponentially attracting, positively invariant
Lipschitz manifold. The principle idea is to determine a map $\Phi:
\mathbb{V} \to \mathbb{W}$ such that we can rewrite equation \reff{absKS} in
the low-dimensional form, \bsplitl{
\partial_t v + {\rm PB}(v+\Phi(v),v+\Phi(v))
    + {\rm AP}v
    =0\,,
}{InvMan} where ${\rm P}\,:\,\mathbb{H}\to \mathbb{V}$ and ${\rm
Q}:= ({\rm I}-{\rm P})\,:\mathbb{H}\to\mathbb{W}$ are projections
onto the orthogonal subspaces $\mathbb{V}$ and $\mathbb{W}$ such
that $\mathbb{H}=\mathbb{V}\oplus\mathbb{W}$. A strategy to
determine $\Phi$ in general Galerkin spaces is, for example,
suggested in \cite{Foias1988} for the KS equation. The RG approach
performed here can also be understood as a formal and feasible
procedure to derive an asymptotic invariant manifold, see
\reff{3.14} and \reff{2.29}.
\new{Further analytical results are the characterization of a global attracting
set for the Kuramoto-Sivashinsky equation by the so-called background flow method \cite{Collet1993,Goodman1994} and via a ``capillary Burgers equation'' \cite{Otto2009}, where the
latter also forms the best known bound in this context. In \cite{Akrivis2013,Collet1993a},
the analyticity of solutions is studied.}

\medskip

{\bf Open questions and answers to the classical separation
\reff{SloFas}:} \emph{ (i) Is the splitting \reff{SloFas} and the
approximation of $u$ by $u^\epsilon$ valid and in which sense?} This
question is often not answered in the literature where from the
outset a separation \reff{SloFas} is assumed, see
\cite{Boghosian1999,CHOW1995} for instance. These studies
heuristically motivate a timescale mediation (or separation into
slow and fast scales) of the form \reff{SloFas}. The present work
aims to provide a rigorous foundation in Theorem \ref{thm:ErEs} by
the following estimate, \bsplitl{ \N{u-v^\epsilon}{L^2({\cg
P}_\alpha)}^2(T)
    \leq C\epsilon^2
            +{\rm exp}\brkts{C T}\brkts{
                \epsilon^{1/4}
                +\epsilon 
            }
        \,.
}{MaErEs}
If we suppose that $u$ and $u^\epsilon$ satisfy a Gevrey regularity characterized by a
parameter $\sigma>0$, then we can improve \reff{MaErEs} in the following way, 
\bsplitl{
\N{u-v^\epsilon}{L^2({\cg P}_\alpha)}^2(T)
    & \leq
    C\epsilon^2
            +{\rm exp}\brkts{C T}\brkts{
                \epsilon^{1/4}{\rm exp}\brkts{-\frac{\sigma}{\epsilon^{1/4}}}
                + \epsilon 
            }
    \biggr)\,.
}{MaGsEst} It should be pointed out that these estimates also
account for the reduction to the slow degrees of freedom
$v^\epsilon$ too and not only for the error between $u$ and
$u^\epsilon$.

\emph{ (ii) How can we account for the fast degrees of freedom
$w^\epsilon$ in an equation for the slow modes $v^\epsilon$ only?}
For this purpose we apply an abstract RG approach extended to
general multiscale problems, see
\cite{Chen1996,Temam1999,Moise2001}. The RG method was first
introduced in quantum field theory as a tool to perform scale
transformations. The method then became popular with Wilson's work
on the Kondo problem \cite{Wilson1975}. It can formally provide the
separation \reff{SloFas}. This means, we first obtain an
approximation for $v^\epsilon$ of the form \bsplitl{
\partial_t{v}^\epsilon
    +{\rm A}_v{v}^\epsilon
    +{\rm P}_N{\rm B}({v}^\epsilon,{v}^\epsilon)
    = -\epsilon {\rm G}_\epsilon(U(t),v^\epsilon)\,,
}{3.14} where $\epsilon G_\epsilon(U,v^\epsilon)$ is a perturbation
``force" originating from the renormalization method and $U=
V+W$ is a solution of the RG equations \bsplitl{
\partial_t V
    +{\rm A}_vV
    +{\rm P}_N{\rm B}(V,V)
    & = 0\,,
\\
\partial_t W
    + {\rm Q}_N{\rm B}_1(V,W)
    & = 0\,,
}{2.29} where $V={\rm P}_N U$ and $W=({\rm I}-{\rm P}_N)U=:{\rm
Q}_NU$ are projections onto the normalized slow and fast manifolds,
respectively. Since we can analytically solve for $W$, we end up
with an equation for the slow variable $v^\epsilon$ only. The above
estimates \reff{MaErEs} and \reff{MaGsEst} then make the reduction
\reff{3.14} rigorous. Moreover, equation \reff{2.29}$_2$ can be
interpreted as the map $\Phi^\epsilon(v^\epsilon)$ onto the
asymptotic invariant manifold.

It should also be pointed out that at this stage the RG
approximation \reff{3.14} alone is not satisfactory since the fast
variable $W$ contained in $U$ is of infinite dimension and hence can
not entirely be resolved numerically. We give an answer to this
problem after the last question (iii) by the principle of maximum entropy. Moreover, question (ii) is of particular
relevance here, since the fast modes prevent the existence of a
canonical invariant measure. Such a measure \new{makes} classical reduction methods
such as  Mori-Zwanzig and optimal prediction \new{more feasible and does not require the
choice of a less physically founded non-invariant measure}, see also question (iii) where
a different methodology is proposed.

\emph{(iii) What kind of information do we need to carry over from
the (infinite dimensional) fast degrees of freedom to the (finite
dimensional) slow ones and how?} To this end, we derive a stochastic
evolution equation for the resolved (slow) variable by properly
including necessary information from the unresolved (fast) variable
by a maximum information entropy principle introduced in
\cite{Jaynes1957a,Jaynes1957,Rosenkrantz1989}. This principle does
not require statistical data to define all Fourier modes. It turns
out that the asymptotic behavior in time of a weighted variance of
the fast modes is sufficient. The necessity of such a strong
assumption relies on the fact the gKS equation does not have an
infinite-dimensional invariant measure and that we only account for
spatial randomness. Via this entropy principle (Theorem
\ref{thm:PrDe}) we then conclude that the Fourier modes of the fast
variable $W$ in \reff{2.29}$_2$ are Gaussian distributed with zero
mean. Hence, we rigorously obtain a noisy gKS equation by applying
the random variable $U=v^\epsilon+W$ in the deterministic equation
for the slow variable \reff{3.14}. Herewith, our analysis explains
how to rigorously add a random force to the gKS equation.
Furthermore, our derivation further shows that the induced noise
accounts for the unresolved degrees of freedom and hence becomes
less important for an increasing number of grid points in
computations.

The approach proposed here provides an alternative to the
Mori-Zwanzig formalism \cite{Mori1965,Zwanzig1961,Zwanzig1973} which advantageously
makes use of a Hamiltonian \cite{Mori1965,Zwanzig1961} or extended Hamiltonian structure \cite{Zwanzig1973}. Mori-Zwanzig techniques
and related optimal
prediction methods \cite{Chorin1998} 
generally rely on a canonical probability distribution (invariant measure)
which exists naturally for Hamiltonian systems. \new{In principle, one can
also apply these techniques to systems that lack an invariant measure.
However, the methodology becomes much more involved in such situations and it
is not clear how to choose the required non-invariant measure unlike with
systems with a canonical invariant measure.} The canonical probability
density for a Hamiltonian $H(u)$ is $\rho(u):=Z^{-1}{\rm exp}\brkts{-\beta
H(u)}$, where $\beta$ is the inverse temperature and $Z$ a normalization
constant referred to as the partition function. The Mori-Zwanzig formalism
then is based on a projection operator ${\rm P}$ that projects functions in
$L^2$ onto a subspace that only depends on the resolved degrees of freedom.
With respect to the canonical density $\rho$ such a projection operator ${\rm
P}$ can be defined by the conditional expectation \bsplitl{ \ebrkts{{\rm
P}f}(v)
    :=\mathbb{E}\ebrkts{f\,\bigl|\,v} =\frac{\int f(v,w)\rho(v,w)\,dw}{\int\rho(v,w)\,dw}\,,
}{PrOp} where $f\in L^2$ and $v$ is the resolved and $w$ the
unresolved variable. The projection ${\rm P}$ and Dyson's formula
for evolution operators then provide an equation for
the resolved modes $v$ only. Moreover, \reff{PrOp} is the
conditional expectation of $f$ given $v$ and hence is the best least
square approximation of $f$ by a function of $v$. Therefore, the
projection ${\rm P}$ guarantees optimality which is the key idea in
the optimal prediction method. However, neither a Hamiltonian
structure nor an invariant measure exists for the gKS equation. Therefore, it is not obvious how to derive
standard optimality statements relying on a conditional probability
argument \cite{Stinis}. In contrast to such a conditional
probability approach, we achieve optimality in the sense of maximum
information entropy. However, we remark that one can also define
other projections than \reff{PrOp}.

\medskip

The purpose of the present article is threefold: 1. To reliably perform a (stochastic) mode reduction for
the full gKS equation in contrast to \cite{Stinis} where a truncated
problem is studied. The principal idea is based on an abstract RG
approach, as emphasized earlier. We derive error estimates (Theorem
\ref{thm:ErEs}) for this reduction and hence provide rigorous
support for the heuristic motivation of a noisy, low dimensional
approximation
deducted in \cite{CHOW1995} by the standard RG method in physics \cite{Wilson1975};\\
2. To rigorously support Stinis' assumption of Gaussian distributed
Fourier modes \cite{Stinis}. To this end, we derive a probability
distribution (Theorem \ref{thm:PrDe}) for the fast modes by the
principle of maximum
information entropy.\\
3. The findings in 1 and 2 form the bases for a new stochastic mode
reduction strategy. We are able to reduce the fast variable by an
equation for the slow variable only. The information of the fast
modes enters as a random variable $W$ via a force term into the slow
mode equations. We are not aware of any previous work that utilizes
the RG method in the context of stochastic mode reduction.

\medskip

We introduce basic notation and well-known
results in Section \ref{sec:Not}. A formal derivation of an RG
equation for the gKS equation follows in Section \ref{sec:RGM}. In
Section \ref{sec:Cmp} we obtain error estimates to rigorously verify
the approximation derived in Section \ref{sec:RGM}. In Section
\ref{sec:StMoRe} we reduce the fast modes by a mode reduction
strategy based on the maximum information entropy principle.
Finally, in Section~\ref{sec:Cocl} we close with conclusions and
perspectives.
\subsection{The gKS equation}\label{sec:KS}
The KS equation is a paradigmatic model for the study of
low-dimensional spatio-temporal chaos or weak/dissipative turbulence
as defined by Manneville~\cite{Manneville}. This type of turbulence
is often characterized by formation of clearly identifiable
localized coherent structures in what appears to be a randomly
disturbed system, as is e.g. the case with Rayleigh-B\'enard
convection~\cite{Shraiman1986}. The KS equation was first proposed
as a model for pattern formation in reaction-diffusion systems by
Kuramoto \cite{Kuramoto1976}. Its derivation is based on a
generalized time-dependent Ginzburg-Landau equation. Sivashinsky
\cite{Sivashinsky1979} derived the KS equation as an asymptotic
approximation of a diffusional-thermal flame model. The equation
also describes small-amplitude waves on the surface of a thin film
flowing down a planar inclined wall
(e.g.~\cite{Homsy74,Kalliadasis2011}).

With the addition of the dispersive term, $u_{xxx}$, the KS equation becomes
the gKS equation. Like the KS equation, it has been reported for a wide
variety of systems, from plasma waves with dispersion due to finite ion
banana width~\cite{Cohen76} to a thin film flowing down a planar wall for
near-critical conditions (e.g. ~\cite{Saprykin2005,Kalliadasis2011}). The
studies in~\cite{Tseluiko2009,Tseluiko2010,Tseluiko2010b,Tseluiko2012} have
developed a coherent-structure theory for the interaction of the
solitary-pulse solutions of the gKS equation.
In~\cite{Tseluiko2009,Tseluiko2010} the theory was shown to be in agreement
with experiments using a thin film coating a vertical fiber, another
hydrodynamic system where the gKS equation can be applicable.

The well-posedness of \reff{absKS} is established for example in
\cite{Tadmor1986} in the class of generalized Burgers equations
which consist of a quadratic nonlinearity and arbitrary linear
parabolic part. The article \cite{LARKIN2004} verifies solvability of the gKS equation
in bounded domains and studies its limit towards the Korteweg-de Vries equation. In the
context of long-time and large-space considerations, there are
recent analytical attempts to verify an ``equipartition principle"
in the power spectrum of periodic solutions by deriving bounds on
their space average of $\av{u}$ and certain derivatives of it, see
\cite{Giacomelli2005,Otto2009}. Such a spectral characterization is
reminiscent of white noise.

An interesting work that applies the optimal prediction to the KS
equation is that of Stinis~\cite{Stinis}. Since this approach
requires a non-invariant measure, the author constructs a Gibbs measure for the required initial distribution
through inference from empirical data (obtained by a computational
approach). This allows then to define the conditional expectation
providing optimality by an orthogonal projection of the unresolved
modes to the resolved ones. However, this approach already assumes a
Gaussian distribution from the outset. For this strategy, one also
needs to work with the truncated KS equation. Sufficient numerical
data are then required in advance for a reliable construction of an
initial distribution.

\subsection{Notation}
\label{sec:Not}
Functions $u\in H^s({\cg P}_\alpha)$ for $s\geq 1$ can be represented by their Fourier
series,
\bsplitl{
u(x)
	= \sum_{k\in\mathbb{Z}}u_k{\rm exp}\brkts{i\frac{k}{\alpha}x}\,,
	\quad
	\ol{u_k}
	= u_{-k}\,,
}{FSE}
where $H^s$ denotes here the usual periodic Sobolev space with finite norm,
\bsplitl{
\N{f}{H^s}^2
	:= \sum_{k\in\mathbb{Z}}(1+\av{k}^2)^s\av{\hat{f}(k)}^2\,.
}{HsNo}
Furthermore, the square root of the
latter quantity is a norm on $H^s({\cal P}_\alpha)$ equivalent to the usual one. We denote for $s\geq 0$,
\bsplitl{
\dot{H}^s({\cal P}_\alpha)
	:= \brcs{
		u\in H^s({\cal P}_\alpha)\,\biggr|\,\int_{{\cal P}_\alpha} u\,dx=0
	}\,.
}{AvH1}
The subspace of $\dot{H}^s({\cg P}_\alpha)$ spanned by the set,
\bsplitl{
\brcs{
	{\rm e}^{i\frac{k}{\alpha}x}\,\bigr|\, k\in\mathbb{Z},\,, -N\leq k\leq N
	}
}{Span}
is denoted by $H_N^s$.
 For a given integer $N$ we define the
projections $v:={\rm P}_Nu$ and $w:={\rm Q}_Nu:=({\rm I}-{\rm P}_N)u$ by,
\bsplitl{
v = {\rm P}_Nu
	= \sum_{\av{k}\leq N}u_k{\rm exp}\brkts{i\frac{k}{\alpha}x}\,,&
\\
w = {\rm Q}_Nu
	= \sum_{\av{k}>N}u_k{\rm exp}\brkts{i\frac{k}{\alpha}x}\,.&
}{PQ} Let us mention that the gKS equation preserves mass as already
noted in Section~\ref{sec:Intr}, i.e., \bsplitl{ \int_{{\cg
P}_\alpha} u\,dx = u_0\,, }{MaCoL} where $u_0$ is the zero-th
Fourier mode. We remark that ${\rm P}_N$ is an orthogonal projection
with respect to $H_N^s$, that means, \bsplitl{ \int_{{\cg
P}_\alpha}({\rm P}_N u - u)\phi\, dx
	= 0
	\qquad\textrm{for all }\phi\in H^s_N\,.
}{Orthog}
The projection ${\rm P}_N$ enjoys the following well-known property \cite{Jackson1930,Maday1988}, i.e., for $k\geq s$, $k\geq 0$ it holds,
\bsplitl{
\N{u-{\rm P}_Nu}{H^s}
	\leq CN^{s-k}\N{u}{H^k}
	\qquad\textrm{for all }u\in H^k_N({\cg P}_\alpha)\,.
}{IntEst}

Next, we introduce Gevrey spaces. For $\sigma\geq 0$
and $s\geq 0$ we say that a function $f$ is in the Gevrey space
$G_{\sigma,s}$ if and only if
\bsplitl{
\N{f}{G_{\sigma,s}}^2
	:=\sum_{k\in\mathbb{Z}}\brkts{
		1+\av{k}^2
	}^s{\rm exp}\brkts{
		2\sigma\sqrt{1+\av{k}^2}
	}\av{f_k}^2
	<\infty\,,
}{GN}
where $f_k$ denote the Fourier coefficients of $f$. Note that if $\sigma=0$, then
$H^s=G_{0,s}$. Moreover, it can be readily be proved, see \cite{Kalisch2007}, that for
$u\in G_{\sigma,s}$ the following inequality holds,
\bsplitl{
\N{u-{\rm P}_Nu}{H^s}
	\leq N^{s-k}{\rm exp}\brkts{
		-\sigma N
	}\N{u}{G_{\sigma,k}}\,.
}{Gint}

\section{Formal derivation of a reduced gKS equation}
\label{sec:RGM}

As noted in Section~\ref{sec:Intr}, we adapt RG
approaches~\cite{Temam1999,Moise2001} to the gKS equation.

\subsection{Projections into fast and slow equations}\label{sec:prEqs}
We apply the projections ${\rm P}_N$ and ${\rm Q}_N$ defined in \reff{PQ} to equation \reff{absKS} and obtain the following coupled system for $v$ and $w$,
\bsplitl{
\partial_tv
	+{\rm P}_N{\rm B}(v+w)
	+{\rm A}_vv
	=0\,,
	&\qquad\textrm{where }{\rm A}_v={\rm P}_N{\rm A}={\rm AP}_N\,,
\\
\partial_tw
	+{\rm Q}_N{\rm B}(v+w)
	+{\rm A}_ww
	=0\,,
	&\qquad\textrm{where }{\rm A}_w={\rm Q}_N{\rm A}={\rm A}{\rm Q}_N\,.
}{SlFa}
We define $\epsilon=\frac{1}{N^4}$, where $N$ is large enough (see error estimates, i.e. Theorem \ref{thm:ErEs},) and set,
\bsplitl{
\Av
	= {\rm A}_v
	&\qquad\textrm{on }{\rm P}_N\dot{H}^s=H^s_N\,,
\\
\Aw
	= \epsilon {\rm A}_w
	= \frac{{\rm A}_w}{N^4}
	&\qquad\textrm{on }{\rm Q}_N\dot{H}^s=\dot{H}^s\setminus H^s_N\,.
}{alpha}
The eigenvectors of $\Av$ are the functions ${\rm exp}\brkts{i\frac{k}{\alpha}x}$, $k\in\mathbb{Z}$,
$\av{k}\leq N$ with eigenvalues,
\bsplit{
\rho^v_k
	:= -\nu\av{\frac{k}{\alpha}}^2
	-i\delta\brkts{\frac{k}{\alpha}}^3
	+\kappa\av{\frac{k}{\alpha}}^4\,.
}
Correspondingly, the eigenvectors of $\Aw$ are the functions ${\rm exp}\brkts{i\frac{k}{\alpha}x}$, $k\in \mathbb{Z}$, $\av{k}>N$
with eigenvalues,
\bsplitl{
\rho^w_k
	:= \frac{1}{N^4}\brkts{
	-\nu\av{\frac{k}{\alpha}}^2
	-i\delta\brkts{\frac{k}{\alpha}}^3
	+\kappa\av{\frac{k}{\alpha}}^4
	}\,.
}{EV}

\begin{rem}\label{rem:1} The RG method is formally applied here as if these
operators $\tilde{A}_v$ and $\tilde{A}_w$ were independent of $\epsilon$. This technical step of
scaling the linear operator and its subsequent treatment is part of
the abstract RG approach introduced in \cite{Temam1999,Moise2001} in
the context of fluid dynamics.
\end{rem}

We can now rewrite \reff{SlFa} by,
\bsplitl{
\partial_t v
	+ \Av v
	+ \PB{v+w}
	=0\,,
&\\
\partial_t w
	+\frac{1}{\epsilon}\Aw w
	+\QB{v+w}
	= 0\,.
}{epsSlFa}
For convenience we additionally define,
\bsplitl{
u = \brkts{\begin{matrix} v \\ w \end{matrix}}\,, \quad {\rm L} = \brkts{\begin{matrix} 0 \\ \Aw \end{matrix}}\,,
\quad {\cal A} = \brkts{\begin{matrix} \Av \\ 0 \end{matrix}}\,,
\quad {\rm F}(u) = \brkts{\begin{matrix} -\PB{v+w} \\ -\QB{v+w} \end{matrix}}\,,
}{CvNt}
and hence rewrite \reff{epsSlFa} in the following compact way,
\bsplitl{
\partial_t u
	+\frac{1}{\epsilon}{\rm L}u
	+{\cal A}u
	= {\rm F}(u)\,.
}{epsKS}

For the subsequent RG analysis we introduce the fast time scale
$s=\frac{t}{\epsilon}$, and we define $\tilde{u}(s)=u(\epsilon s)$.
We set $\tilde{v}(s)={\rm P}\tilde{u}(s)$, $\tilde{w}(s)={\rm
Q}\tilde{u}(s)$. In this variables \reff{epsSlFa} becomes, \bsplitl{
\partial_s\tilde{v}
	+ \epsilon\Av\tilde{v}
	+\epsilon\PB{\tilde{v}+\tilde{w}}
	=0\,,
&\\
\partial_s\tilde{w}
	+ \Aw\tilde{w}
	+\epsilon\QB{\tilde{v}+\tilde{w}}
	=0\,,
&	
}{FaKS}
or \reff{epsKS},
\bsplitl{
\partial_s\tilde{u}
	+{\rm L}\tilde{u}
	+\epsilon{\cal A}\tilde{u}
	= \epsilon {\rm F}(\tilde{u})\,.
}{FaepsKS}

\subsection{Perturbation expansion: The RG equation}\label{sec:PeEx} We now formally apply
the RG method and additionally
omit the dependence of ${\rm L}$ and ${\cal A}$ on $N$ as in \cite{Temam1999,Moise2001}.
For simplicity, we also assume 
that either (i) $L$, $\nu$, $\kappa$ are not proportional to $\pi$ or that (ii) $N^2\geq \frac{8}{7}\frac{\alpha^2\nu}{\kappa}$ where $N$ denotes the largest Fourier mode in the Galerkin approximation.

We make the ansatz of a naive perturbation expansion, \bsplitl{
\tilde{u}^\epsilon
	= \tilde{u}^0
	+ \epsilon\tilde{u}^1
	+ \epsilon^2\tilde{u}^2
	+ \dots\,.
}{NaEx}
for $\tilde{u}$ in \reff{FaepsKS}. After substituting \reff{NaEx} into \reff{FaepsKS} we formally obtain the following sequence of problems,
\bsplitl{
\partial_s\tilde{u}^0
	+ {\rm L}\tilde{u}^0
	& = 0\,,
\\
\partial_s\tilde{u}^1
	+{\rm L}\tilde{u}^1
	&= {\rm F}(\tilde{u}^0)
	-{\cal A}\tilde{u}^0\,,
}{OrEqs}
and so on.

Formally, the solution of \reff{OrEqs}$_1$ for the initial condition $\tilde{u}^0(0)=u_0$ is,
\bsplitl{
\tilde{u}^0(s)
	= {\rm exp}\brkts{-{\rm L}s}u_0\,.
}{0InCo}
Equation \reff{0InCo} can be equivalently written by,
\bsplitl{
\tilde{v}^0(s)
	= v_0\,,
&\\
\tilde{w}^0(s)
	= {\rm exp}\brkts{-\Aw s}w_0\,.
}{0InCo1}
We solve equation \reff{OrEqs}$_2$ with the variation of constants formula,
\bsplitl{
\tilde{u}^1(s)
	= {\rm exp}\brkts{-{\rm L}s}
	\int_0^s{\rm exp}({\rm L}\sigma)
		\ebrkts{
			{\rm F}({\rm exp}(-{\rm L}\sigma)u_0)-{\cg A}{\rm exp}(-{\rm L}\sigma)u_0
		}\,d\sigma\,, }{u1} where $\tilde{u}^1(0)=0$, since we are
interested in approximations up to $O(\epsilon)$ such that $\tilde{u}^1(0)$ is irrelevant and can be taken to
be zero, see \cite{Moise2001}. We note that 
${\cg
A}{\rm exp}(-{\rm L}\sigma)={\cg A}$ and we decompose the rest of
the integrand in \reff{u1} as, \bsplitl{ {\rm exp}({\rm
L}\sigma){\rm F}({\rm exp}(-{\rm L}\sigma)u_0)
	- {\cg A}v_0
	=: {\rm F}_R(u_0)
	+ \tilde{\rm F}_{NR}(\sigma,u_0)\,, }{F_R} where ${\rm
F}_R(u_0)$ represents the part  independent of $\sigma$ on the left
hand side of \reff{F_R} and $\tilde{\rm F}_{NR}$ the rest. Using
standard RG terminology, we refer to ${\rm F}_R$ as the ``resonant"
and $\tilde{\rm F}_{NR}$ as the ``non-resonant" term.

Using \reff{0InCo}, \reff{u1}, and \reff{F_R} in \reff{NaEx} provides the following Duhamel's form of
the formal perturbation expansion for $\tilde{u}=\tilde{u}^\epsilon$,
\bsplitl{
\tilde{u}^\epsilon(s)
	= {\rm exp}(-{\rm L}s)\brkts{
		u_0
		+ \epsilon s {\rm F}_R(u_0)
		+\epsilon \int_0^s\tilde{\rm F}_{NR}(\sigma,u_0)\,d\sigma
	} + {\cg O}(\epsilon^2)\,. }{uEps} The key idea is now to remove
the secular term $\epsilon s{\rm F}_R(u_0)$ which grows in time. To
this end we define the ``renormalized function"
$\tilde{U}=\tilde{U}(s)$ as the solution of, \bsplitl{
\partial_s\tilde{U}
	& = \epsilon {\rm F}_R(\tilde{U})\,,
\\
\tilde{U}(0)
	& = u_0\,.
}{RGeq}
The equation for the slow variable $U(t)=\tilde{U}(t/\epsilon)$ correspondingly satisfies,
\bsplitl{
\partial_tU
	& = {\rm F}_R(U)\,,
\\
U(0)
	& = u_0\,.	
}{RGeq1}
Let us derive the explicit form of the RG equation for our problem. With the expressions
for ${\rm L}$ and ${\rm F}$, and the identity $u_0 = v_0 + w_0$, we get,
\bsplitl{
{\rm exp}({\rm L}\sigma)
	{\rm F}({\rm exp}(-{\rm L}\sigma)u_0)
	= {\rm exp}({\rm L}\sigma)
	\brkts{
	\begin{matrix} -\PB{v_0+{\rm exp}(-{\rm L}\sigma)w_0}
	\\
	-\QB{v_0+{\rm exp}(-{\rm L}\sigma)w_0}
	\end{matrix}
	}
	\,.
}{eFe}
Next we identify the resonant terms, i.e., ${\rm F}_R(u_0)$. With the Fourier series expansion,
\bsplitl{
\phi(x)
	= \sum_{k\in\mathbb{Z}} {\rm
exp}\brkts{i\frac{k}{\alpha}x}\phi_k\,, }{FS} we have, \bsplitl{
B(\phi,\psi)
	& = i\lambda \sum_{k\in\mathbb{Z}}{\rm exp}\brkts{i\frac{k}{\alpha}x}\phi_k
	\sum_{l\in\mathbb{Z}}{\rm exp}\brkts{i\frac{l}{\alpha}x}\frac{l}{\alpha}\psi_l
\\
	& = i\lambda \sum_{j\in \mathbb{Z}} {\rm exp}\brkts{i\frac{j}{\alpha}x}
	\sum_{k+l=j}\brkts{\phi_k\frac{ l}{\alpha}}\psi_l
	\,.
}{B}
As a consequence, we end up with the expressions,
\bsplitl{
& \QB{v_0,{\rm exp}(-{\rm L}\sigma)w_0}
	 = i\lambda \sum_{\av{j}>N} {\rm exp}\brkts{i\frac{j}{\alpha}x}
	\sum_{\substack{k+l=j \\
		\av{k}\leq N <\av{l}}
		}
		\brkts{v_{0k}\frac{l}{\alpha}}{\rm exp}(-\sigma\rho^w_l)w_{\new{0}l}\,,
\\
& \QB{\rm exp(-{\rm L}\sigma)w_0,{\rm exp}(-{\rm L}\sigma)w_0}
	 = i\lambda \sum_{\av{j}>N} {\rm exp}\brkts{i\frac{j}{\alpha}x}
\qquad\qquad\qquad
\\&\qquad\qquad\qquad\qquad
	\sum_{\substack{k+l=j \\
		\av{k},\av{l}> N }
		}
		\brkts{{\rm exp}(-\sigma\rho^w_k)w_{0k}\frac{l}{\alpha}}{\rm
exp}(-\sigma\rho^w_l)w_{\new{0}l}
	\,. 
}{Bs} 
The resonant terms in the first
sum are the terms for which $\rho^w_l=\rho^w_j$ holds. 
\new{We note that one also needs to look at the skew-symmetric bilinear 
form $\QB{{\rm exp}(-{\rm L}\sigma)w_0,v_0}$ which leads to the same resonance 
condition,} this means,
$\brkts{-\av{\frac{l}{\alpha}}^2-i\delta\brkts{\frac{l}{\alpha}}^3+\kappa\av{\frac{l}{\alpha}}^4}
=
\brkts{-\av{\frac{j}{\alpha}}^2-i\delta\brkts{\frac{j}{\alpha}}^3+\kappa\av{\frac{j}{\alpha}}^4}$.
Since $\nu,\delta,\kappa>0$, the following set characterizes the
resonant indices, \bsplitl{ R_1(j) & :=
	\brcs{(k,l)\,\bigr|\, k=0,\,j=l,\,\av{l}>N}\,.
}{LeqK}
The condition $\rho^w_k+\rho^w_l=\rho^w_j$ characterizes the resonant terms in the second sum
of \reff{Bs}, i.e.,
$\brkts{\frac{k}{\alpha}}^n+\brkts{\frac{l}{\alpha}}^n=\brkts{\frac{j}{\alpha}}^n$ for $n=2,3,4$ needs to hold at the same time. 
Assuming that this condition holds for $n=2$, we immediately obtain an additional requirement $\av{\frac{j}{\alpha}}^4=\av{\frac{k}{\alpha}}^4+\av{\frac{l}{\alpha}}^4+2\av{\frac{k}{\alpha}}^2\av{\frac{l}{\alpha}}^2$, which holds with respect to the set of resonant indices defined by,
\bsplitl{
R_2(j) :=
	\brcs{(k,l)\,\bigr|\, k=0, l=j,\,\av{k},\av{l}>N}\cup\brcs{(k,l)\,\bigr|\, l=0,\, k=j,\,\av{k},\av{l}>N}
	=\emptyset\,,
}{KeqLeq0}
since $\av{k},\av{l}>N$. For a rigorous and detailed proof we refer to the Appendix. We immediately recognize that \reff{KeqLeq0} also 
justifies our assumption on the case $n=2$ above.

These considerations determine the resonant part of ${\rm F}$ by,
\bsplitl{
{\rm F}_R(u_0)
	= \ebrkts{\begin{matrix}
	-\PB{v_0} - \Av v_0
	\\
	-{\rm Q}_N{\rm B}_1(v_0,w_0)
	\end{matrix}}\,,
}{FR}
where ${\rm B}_1$ is 
given by its Fourier series expansions for the corresponding index set $R_1(j)$, i.e., 
\bsplitl{
{\rm Q}_N{\rm B}_1(v_0,w_0)
	& = 2i\lambda\sum_{\av{j}>N}{\rm e}^{i\frac{j}{\alpha}x}
	\brkts{v_{00}\frac{j}{\alpha}}w_{0j}
	\,.
}{QB1}
Equation \reff{F_R} and the above consideration give the non-resonant term by,
\bsplitl{
\tilde{\rm F}_{NR}(\sigma,u_0)
	= \ebrkts{\begin{matrix}
		-\new{\PB{v_0+{\rm e}^{-\Aw\sigma}w_0
		}}
		+\new{\PB{v_0}}
		+\tilde{\rm A}_vv_0
	\\
	-{\rm Q}_N\tilde{\rm B}_1(v_0,w_0)
	-\new{{\rm Q}_N\tilde{\rm B}_2(w_0)}
	\end{matrix}}\,,
}{FNR}
where ${\rm Q}_N\tilde{B}_1$ and ${\rm Q}_N\tilde{\rm B}_2$ are defined by their Fourier series
expansions,
\bsplitl{
{\rm Q}_N\tilde{\rm B}_1(v_0,w_0)
	& = i\lambda\sum_{\av{j}>N}{\rm e}^{i\frac{j}{\alpha}x}
		\sum_{
		\substack{
			k+l=j\\
			\av{l}\neq\av{j}\\
			\av{k}\leq N<\av{l}
		}
		}\brkts{
		\brkts{v_{0k}\frac{j}{\alpha}}w_{0l}
		+\brkts{w_{0l}\frac{j}{\alpha}}v_{0k}
		}{\rm e}^{(\rho^w_j-\rho^w_l)\sigma}\,,
\\
{\rm Q}_N\tilde{\rm B}_2(w_0)
	& = i\lambda\sum_{\av{j}>N}{\rm e}^{i\frac{j}{\alpha}x}
		\sum_{
		\substack{
			k+l=j\\
			|k/\alpha|^n+|l/\alpha|^n\neq |j/\alpha|^n\,\, {\rm for}\,\, n=2,3,4\\
			\av{k},\av{l}>N
		}
		}
		\brkts{w_{0k}\frac{j}{\alpha}}w_{0l}
		{\rm e}^{(\rho^w_j-\rho^w_k-\rho^w_l)\sigma}\,.
}{QB1QB2}
With \reff{FR} the RG equation for our problem is in the fast time scale,
\bsplitl{
\partial_s\tilde{V}
	+\epsilon\Av\tilde{V}
	+\new{\epsilon\PB{\tilde{V}}} 
	& = 0\,,
\\
\partial_s \tilde{W}
	+\epsilon{\rm Q}_N{\rm B}_1(\tilde{V},\tilde{W})
	& =0\,,
}{rgeq}
or after rescaling by $t=\epsilon s$, and denoting $V={\rm P}_NU$,
$W={\rm Q}_NU$,
\bsplitl{
\partial_t V
	+{\rm A}_vV
	+\new{{\rm P}_N{\rm B}(V)}
	& = 0\,,
\\
\partial_t W
	+ {\rm Q}_N{\rm B}_1(V,W)
	& = 0\,.
}{rgeq1}

\begin{rem}\label{rem:2} 1) The above considerations for the resonant and non-resonant terms can easily be extended to \new{situations where we replace the linear 
spatial differential operator ${\rm A}$ with} 
pseudodifferential operators \new{${\rm P}(\partial/\partial x)$} with symbol $p(\xi)$ of the
form, \bsplitl{ {\rm Re}\, p(i\xi) \geq c\av{\xi}^{\nu}\,,
	\qquad\av{\xi}\to\infty\,,
}{pxi}
where $\nu>3/2$. \new{The requirement \reff{pxi} on the ${\rm P}(\partial/\partial x)$ guarantees 
the well-posedness (of such generalized Burgers equations) \cite{Tadmor1986}.}
One only needs to adapt the sets for the resonant indices, see \reff{LeqK}
and \reff{KeqLeq0}.\\
2) Note that the $V$-equation in the RG equation \reff{rgeq1} is
simply the Galerkin approximation of the gKS equation \reff{absKS}.
\end{rem}

The special structure of the renormalization equation \reff{rgeq1}$_2$ for the unresolved (fast) variable allows to give an explicit expression for its solution. After rewriting \reff{rgeq1}$_2$ by
\bsplitl{
\partial_t W_j(t)
	+ 2i\lambda\frac{j}{\alpha}V_0(t)W_j(t)=0\,,
}{Wj}
where $V_0(t)=const.$ due to conservation of mass \reff{MaCoL}, we immediately obtain the solution,
\bsplitl{
W_j(t)
	= c^j_We^{i2\lambda\frac{j}{\alpha}V_0t}\,,
	\qquad
	c^j_W := W_j(0)\,.
}{Wj1}
With \reff{Wj1} the solution of \reff{rgeq1}$_2$ becomes,
\bsplitl{
W(x,t)
	= \sum_{\av{j}>N}c^j_We^{i\frac{j}{\alpha}\brkts{x+2\lambda V_0t}}\,.
}{Wj2}
Equation \reff{Wj2} shows that there is no restriction on the definition of the
mass $V_0$. In the context of stochastic mode reduction the situation
is different, see Section \ref{sec:StMoRe}.

\subsection{Construction of approximate/renormalized solutions}\label{sec:ApSo}
In order to define renormalized solutions we have first to determine the non-resonant term
$\tilde{\rm F}_{NR}(\sigma,u_0)$ given by \reff{FNR}. In fact, we are interested in,
\bsplitl{
{\rm F}_{NR}(s,U)
	= \int_0^s\tilde{\rm F}_{NR}(\sigma,U)\,d\sigma\,.
}{FNR1}
Let
\bsplitl{
{\rm PF}_{NR}(s,U)
	& = 2i\lambda\sum_{\av{j}\leq N}{\rm e}^{i\frac{j}{\alpha}x}
		\sum_{\substack{
			k+l=j\\
			\av{k}\leq N<\av{l}
		}}\frac{{\rm e}^{-\rho^w_ls}}{\rho^w_l}V_k\frac{j}{\alpha}W_{l}
\\&\quad
	+ i\lambda\sum_{\av{j}\leq N}{\rm e}^{i\frac{j}{\alpha}x}
		\sum_{\substack{
			k+l=j\\
			\av{k},\av{l}> N
		}}\frac{{\rm e}^{-(\rho^w_k+\rho^w_l)s}}{\rho^w_k+\rho^w_l}
			W_k\frac{j}{\alpha}W_{l}\,,
\\
{\rm QF}_{NR}(s,U)
	& =  -2i\lambda\sum_{\av{j} > N}{\rm e}^{i\frac{j}{\alpha}x}
		\sum_{\substack{
			k+l=j\\
			\av{k}\leq N<\av{l}\\
			\av{l}\neq\av{j}
		}}\frac{{\rm e}^{(\rho^w_j-\rho^w_l)s}-1}{\rho^w_j-\rho^w_l}
			V_k\frac{j}{\alpha}W_{l}
\\&\quad
	- i\lambda\sum_{\av{j}>N}{\rm e}^{i\frac{j}{\alpha}x}
		\sum_{\substack{
			k+l=j\\
			|k/\alpha|^n+|l/\alpha|^n\neq |j/\alpha|^n\,\, {\rm for}\,\, n=2,3,4\\
			\av{k},\av{l}> N
		}}\frac{{\rm e}^{(\rho^w_j-\rho^w_k-\rho^w_l)s}-1}{\rho^w_j-\rho^w_k-\rho^w_l}
			W_k\frac{j}{\alpha}W_{l}\,.
}{PQ_FNR}
Now, we are able to define the approximate solution suggested by the
RG theory. We obtain,
\bsplitl{
{u}^\epsilon(t)
	= {\rm e}^{-{\rm L}\frac{t}{\epsilon}}\brkts{
		U(t)
		+\epsilon {\rm F}_{NR}(t/\epsilon,U(t))
	}\,,
}{ApSo}
or with respect to fast $\ol{w}^\epsilon$ and slow variables $\ol{v}^\epsilon$,
\bsplitl{
{v}^\epsilon
	& ={\rm P}_N{u}^\epsilon
	= V(t) + \epsilon {\rm PF}_{NR}(t/\epsilon,U(t))\,,
\\
{w}^\epsilon
	& = {\rm Q}_N{u}^\epsilon
	= {\rm e}^{-{\rm Q}_N{\rm A}t}\brkts{
	W(t)+\epsilon{\rm QF}_{NR}(t/\epsilon,U(t))
	}\,.
}{ApSo1}
We note that the initial data are defined by,
\bsplitl{
{v}^\epsilon(0)
	& = V(0) + \epsilon{\rm PF}_{NR}(0,U(0)) = v_0+\epsilon{\rm PF}_{NR}(0,u_0)\,,
\\
{w}^\epsilon(0)
	& = W(0) + \epsilon{\rm QF}_{NR}(0,U(0)) = w_0\,.
}{InApSo1}

\section{The renormalized gKS equation and approximation error}\label{sec:Cmp}
After inserting \reff{ApSo} into \reff{absKS} we obtain the
following perturbed gKS equation, \bsplitl{
\partial_t{u}^\epsilon
	+{\rm A}{u}^\epsilon
	+{\rm B}({u}^\epsilon,{u}^\epsilon)
	= -\epsilon {\rm R}_\epsilon(U(t))\,,
}{PeKSeq}
where ${\rm A}$ and ${\rm B}$ are defined by
\reff{absOp} and ${\rm R}_\epsilon$ is given by,
\bsplitl{
{\rm R}_\epsilon
	& = {\rm B}\brkts{{\rm e}^{-{\rm L}\frac{t}{\epsilon}}U(t),{\rm e}^{-{\rm L}\frac{t}{\epsilon}}
		{\rm F}_{NR}(t/\epsilon,U(t))
	}
	+{\rm B}\brkts{{\rm e}^{-{\rm L}\frac{t}{\epsilon}}{\rm F}_{NR}(t/\epsilon,U(t)),{\rm e}^{-{\rm L}\frac{t}{\epsilon}}
		U(t)
	}
\\&\quad
	+{\rm B}\brkts{{\rm e}^{-{\rm L}\frac{t}{\epsilon}}{\rm F}_{NR}(t/\epsilon,U(t)),{\rm e}^{-{\rm L}\frac{t}{\epsilon}}
		{\rm F}_{NR}(t/\epsilon,U(t))
	}
	-{\rm A}{\rm PF}_{NR}(t/\epsilon,U(t))
\\&\quad
	-{\rm e}^{-{\rm L}\frac{t}{\epsilon}}\delta_U{\rm F}_{NR}(t/\epsilon,U(t))\partial_tU\,.
}{Reps}

Next, we study estimates on the approximate solutions ${u}^\epsilon$ of
equation \reff{PeKSeq}. In a first step, we need to investigate the non-resonant
part ${\rm F}_{NR}$ of the approximate solutions.

\medskip

\begin{lem}\label{thm:Fest}
Let $p\geq 2$ and let \new{the initial condition satisfy} $g\in H^{q}({\cg P}_\alpha)$ with $q\geq 4$. Assume that the solution of the RG equation
\reff{rgeq1} satisfies $U(t)\in H^p({\cg P}_\alpha)$ for all $t>0$. For $N$ large enough
there exist
two \new{uniform} constants $c_1$ and $C_2$, where $C_2$ depends on the initial conditions, and $c_1$ depends
only on ${\cg P}_\alpha$, but both independent of $N$, such that the
following estimates are true for all $t>0$,
\bsplitl{
\N{{\rm P}_N{\rm F}_{NR}(t/\epsilon,U(t))}{H^p}
	& \leq C_2 {\rm e}^{-c_1N^4t}\,,
\\
\N{{\rm e}^{-{\rm Q}_NAt}{\rm Q}_N{\rm F}_{NR}(t/\epsilon,U(t))}{H^p}
        & \leq C_2 {\rm e}^{-c_1N^4t}\,.
}{ApEs}
\end{lem}

\medskip

\begin{remark}\label{rem:lem3.1} \new{\emph{(Initial conditions)} We note that the 
regularity assumed above in Lemma \ref{thm:Fest} and in the results below is 
slightly higher since $\g\in H^p({\cg P}_\alpha)$ would be enough. This regularity 
assumption enters via an argument based on Gronwall's inequality.
}
\end{remark}

\medskip

\begin{proof}
$c_1$, and $C_2$ represent generic constants independent of $N$
(or $\epsilon$). We first derive estimate \reff{ApEs}$_1$. With the
expression \reff{PQ_FNR}$_1$ we immediately obtain, \bsplitl{
\N{{\rm P}_N{\rm F}_{NR}(t/\epsilon,U(t))}{H^p}
	\leq \new{C_2} {\rm e}^{-\new{c_1}N^4t}\brkts{
	\N{V\cdot\nabla W}{H^p}
	+\N{W\cdot\nabla V}{H^p}
	+\N{W\cdot\nabla W}{H^p}
	}\,,
}{PFest}
where we used the fact that we have the following bound,
\bsplitl{
\frac{{\rm e}^{-\rho^w_lt/\epsilon}}{\rho^w_l}
	=\frac{{\rm e}^{-(-\nu\av{\frac{l}{\alpha}}^2
			-i\delta\brkts{\frac{l}{\alpha}}^3
			+\kappa\av{\frac{l}{\alpha}}^4)t}}{
	1/N^4(-\nu\av{\frac{l}{\alpha}}^2
                        -i\delta\brkts{\frac{l}{\alpha}}^3
                        +\kappa\av{\frac{l}{\alpha}}^4)}
	\leq
	\new{C_2}{\rm e}^{-\new{c_1}N^4t}\,.
}{ExpEst}
The last inequality follows due to $N< |l|$. The second estimate
\reff{ApEs}$_2$ can be obtained in the same way by using the inequalities
$1-{\rm e}^{-x}\leq x$ for all $x\geq 0$ and
$x{\rm e}^{-x}\leq \frac{1}{\rm e}$ for all $x\geq 0$. We refer the interested reader to \cite{Temam1999}
for a deeper consideration.
\end{proof}

\medskip

The bounds of Lemma \ref{thm:Fest} allow us to control $R_\epsilon$ in the spirit of \cite{Temam1999}.

\medskip

\begin{lem}\label{thm:Rest}
For $N>0$ and \new{for initial conditions} $g\in H^q({\cg P}_\alpha)$ for $q\geq 4$, there
exist two constants $c_1$ and $C_2$ independent of $N$, such that the following
estimate holds true for all $t\geq 0$,
\bsplitl{
\N{R_\epsilon(t)}{L^2}
	\leq C_2{\rm e}^{-c_1N^4t}\,.
}{ApSolEs}
\end{lem}

\begin{proof}
The proof follows in the same way as the proof of Lemma \ref{thm:Fest}. We only
need to take into account the expression of $R_\epsilon$ and apply
Lemma \ref{thm:Fest}.
\end{proof}

\medskip

Subsequently, we write $\Ll{\cdot}$ and $\brkts{\cdot,\cdot}$ for the 
$L^2({\cal P}_\alpha)$-norm and the $L^2({\cal P}_\alpha)$-scalar product, respectively.

\medskip

\begin{lem}\label{lem:ApEsRS}
For $0<T^*<\infty$ and \new{initial conditions} $g$ as in Lemma \ref{thm:Rest}, there exists an $0<\epsilon^*<\infty$ such that for $0\leq \epsilon :=1/N^4\leq
\epsilon^*$ solutions to equation \reff{PeKSeq} satisfy
$u^\epsilon\in L^\infty(0,T^*;L^2({\cg P}_\alpha))\cap L^2(0,T^*;H^2({\cg P}_\alpha))$.
\end{lem}

\medskip

\begin{proof} We give here the elements of the proof for the case  $\kappa> \nu $ and refer to \cite{LARKIN2004,Tadmor1986} where stronger regularity (e.g. $u^{\epsilon}\in L^{\infty}(0;T;H^2(P_{\alpha}))$) and existence results can be found.
We formally test equation
\reff{PeKSeq} with $u^\epsilon$ and using periodicity of ${\cal P}_\alpha$, i.e.,
$\frac{\lambda}{6}\brkts{\partial_x (u^\epsilon)^3,1}=0$, such that,
\bsplitl{
\frac{1}{2}\frac{d}{dt}\Ll{u^\epsilon}^2
	& +(\kappa-\nu)\Ll{\partial_x^2 u^\epsilon}^2
	\leq \frac{\epsilon^2}{2}\Ll{u^\epsilon}^2
	+\frac{1}{2}\Ll{R_\epsilon(U)}^2
\\&
}{ueLe1}
where used the inequality 
$\Ll{\nabla u^\epsilon}^2\leq \Ll{\Delta u^\epsilon}^2$ which holds in the periodic case (see \cite[p. 3]{Tadmor1986}).
After defining
\bsplitl{
\beta
	& := 2
	\epsilon^2
	\,,
\\
\gamma
	& := \frac{1}{2}\Ll{R_\epsilon(U)}^2\,,
}{ag}
we multiply \reff{ueLe1} by ${\rm exp}\brkts{-\int_0^t\beta\,ds}$ such that
\bsplitl{
\frac{1}{2}{\rm exp}\brkts{
		-\beta t
	}\frac{d}{dt}\Ll{u^\epsilon}^2
	\leq
	{\rm exp}\brkts{
		-\beta t
	}\frac{\beta}{2}\Ll{u^\epsilon}^2
	+{\rm exp}\brkts{
		-\beta t
	}\gamma(t)
\,.
}{ueLe2}
Since
\bsplitl{
\frac{d}{dt}\brkts{
		{\rm exp}\brkts{-\beta t}\frac{1}{2}\Ll{u^\epsilon}^2
	}
	= -\beta {\rm exp}\brkts{
		-\beta t
	}\frac{1}{2}\Ll{u^\epsilon}^2
	+{\rm exp}\brkts{
		-\beta t
	}\frac{1}{2}\frac{d}{dt}\Ll{u^\epsilon}\,,
}{prRu}
we can rewrite \reff{ueLe2} as
\bsplitl{
\frac{d}{dt}\brkts{
		{\rm exp}\brkts{-\beta t}
		\frac{1}{2}\Ll{u^\epsilon}^2
	}
	\leq {\rm exp}\brkts{
		-\beta t
	}\gamma(t)\,,
}{ueLe3}
and subsequent integration together with Lemma \ref{thm:Rest} gives,
\bsplitl{
	\Ll{u^\epsilon(T)}^2
	\leq
	C{\rm exp}\brkts{\beta T}
	\int_0^T{\rm exp}\brkts{
		\brkts{\beta -C/\epsilon}t
	}\,dt
	\leq \frac{C}{\beta+C/\epsilon}{\rm exp}\brkts{
		\beta T
	}\,.
}{ueLe4}
For aribtrary $0<T^*<\infty$ we can choose $0\leq \epsilon\leq\epsilon^*:=\frac{1}{{\rm exp}\brkts{\beta T^*}-\beta/CC_\infty}$ where the constant $C_\infty$ is chosen such that,
\bsplitl{
\frac{C}{\beta+C/\epsilon}{\rm exp}\brkts{\beta T^*}
	\leq C_\infty
	< C{\rm exp}\brkts{\beta T^*}/\beta
	\,.
}{ueLe5}
\end{proof}

\medskip

The reduced equation for the resolved (slow) modes $v^\epsilon$ alone follows immediately after using
\reff{ApSo1}$_1$ and \reff{rgeq1}$_1$, i.e., $V(t)=v^\epsilon(t)-\epsilon {\rm PF}_{NR}(t/\epsilon,U(t))$,
\bsplitl{
\partial_t{v}^\epsilon
	+{\rm A}_v{v}^\epsilon
	+{\rm P}_N{\rm B}({v}^\epsilon,{v}^\epsilon)
	= -\epsilon {\rm G}_\epsilon(U(t),v^\epsilon)\,,
}{PeKSeq2}
where the induced force term ${\rm G}_\epsilon$ is defined by,
\bsplitl{
{\rm G}_\epsilon(U(t),v^\epsilon)
	 &:= {\rm P}_N{\rm B}\brkts{{\rm e}^{-{\rm L}\frac{t}{\epsilon}}v^\epsilon,{\rm e}^{-{\rm L}\frac{t}{\epsilon}}
		{\rm PF}_{NR}(t/\epsilon,U(t))
	}
\\&\quad
	+{\rm P}_N{\rm B}\brkts{{\rm e}^{-{\rm L}\frac{t}{\epsilon}}{\rm PF}_{NR}(t/\epsilon,U(t)),{\rm e}^{-{\rm L}\frac{t}{\epsilon}}
		v^\epsilon
	}
\\&\quad
	+{\rm P}_N{\rm B}\brkts{{\rm e}^{-{\rm L}\frac{t}{\epsilon}}{\rm PF}_{NR}(t/\epsilon,U(t)),{\rm e}^{-{\rm L}\frac{t}{\epsilon}}
		{\rm PF}_{NR}(t/\epsilon,U(t))
	}
\\&\quad
	-{\rm A}_v{\rm PF}_{NR}(t/\epsilon,U(t))
	-{\rm e}^{-{\rm L}\frac{t}{\epsilon}}\delta_U{\rm PF}_{NR}(t/\epsilon,U(t))\partial_tU
	\,,
}{Geps}
and $U(t)=V(t)+W(t)=v^\epsilon(t)+W(t)$ is the solution of the RG equation \reff{rgeq1}.

\medskip

\begin{lem}\label{lem:GaApEs}
For $0<T^*<\infty$ and $g$ as in Lemma \ref{thm:Rest}, there exists an $0\leq \epsilon^*<\infty$ such that for $0\leq \epsilon =\frac{1}{N^4}
\leq \epsilon^*$ solutions to \reff{PeKSeq2} satisfy
$v^\epsilon :={\rm P}_N u^\epsilon\in L^\infty(0,T^*;L^2({\cg P}_\alpha))\cap L^2(0,T^*;H^2({\cg P}_\alpha))$.
\end{lem}

\medskip

\begin{proof}
The proof is similar to the proof of Lemma \ref{lem:ApEsRS}.
\end{proof}

The same arguments as those for Lemma \ref{thm:Rest} lead to the following.

\medskip

\begin{lem}\label{lem:Gest}
For $N>0$ and $v_0={\rm P}_Ng\in H^q({\cg P}_\alpha)$ for $q\geq 4$, there
exist two constants $c_1$ and $C_2$ independent of $N$, such that the following
estimate holds true for all $t\geq 0$,
\bsplitl{
\N{G_\epsilon(t)}{L^2}
	\leq C_2{\rm e}^{-c_1N^4t}\,.
}{ReApSolEs}
\end{lem}

\medskip

The following theorem gives
qualitative information about the RG approach by quantifying the
error between \reff{PeKSeq2} and \reff{absKS}.

\medskip

\begin{thm}\label{thm:ErEs} Let $g\in H^4({\cg P}_\alpha)$,  $\epsilon = \frac{1}{N^4}$,  and suppose that
$u, u^\epsilon\in L^\infty(0,T; H^2({\cg P}_\alpha))$. Then, the
difference between the reduced solution $v^\epsilon$ and the exact
solution of the gKS equation \reff{absKS} satisfies the following
error estimate, 
\bsplitl{ \N{u(T)-v^\epsilon(T)}{L^2({\cg P}_\alpha)}^2
	\leq C\epsilon^2
			+{\rm exp}\brkts{C T}\brkts{
				\epsilon^{1/4}
				+\epsilon 
			}
		\,.
}{mErEs}
If we suppose that $u, u^\epsilon\in L^\infty(0,T; G_{\sigma,2}({\cg P}_\alpha))$, then we can
improve \reff{mErEs} in the following way,
\bsplitl{
\N{u(T)-v^\epsilon(T)}{L^2({\cg P}_\alpha)}^2
	& \leq
	C\epsilon^2
			+{\rm exp}\brkts{C T}\brkts{
				\epsilon^{1/4}{\rm exp}\brkts{-\frac{\sigma}{\epsilon^{1/4}}}
				+ \epsilon 
			}
	\biggr)\,.	
}{GsEst}
\end{thm}

\begin{rem}\label{rem:ExpGr}
{\rm 1.} The exponential growth in time is not surprising, see for example
estimate \reff{ueLe4} in the proof of Lemma \ref{lem:ApEsRS}. This
estimate motivates the definition of a new variable
$h(x,t;\eta):={\rm exp}(-\eta t)u(x,t)$. Tadmor verifies in
\cite{Tadmor1986} global existence for such a decayed variable $h$
and a conservative form of the KS equation.\\
{\rm 2.} The assumption $u, u^\epsilon\in L^\infty(0,T; H^2({\cg P}_\alpha))$ is a direct consequence of a priori estimates, which are derived by analogous steps as in the proof of 
Lemma \ref{lem:ApEsRS}, and of imposing initial conditions 
$u_0, u^\epsilon_0\in H^2({\cg P}_\alpha)$.
\end{rem}

\medskip

\begin{proof}
The error $\N{u-v^\epsilon}{L^2({\cg P}_\alpha)}$ can be bounded using the triangle inequality by,
\bsplitl{
\N{u-v^\epsilon}{L^2({\cg P}_\alpha)}
	\leq \N{u-u^\epsilon}{L^2({\cg P}_\alpha)}
	+\N{u^\epsilon-v^\epsilon}{L^2({\cg P}_\alpha)}\,,
}{ErEs1}
where the first term on the right-hand side in \reff{ErEs1} represents the approximation error
from the RG method (\emph{RG error}) and the second term accounts for the
truncation error (\emph{Tr error}). For notational brevity, we introduce the error variables
\bsplitl{
E^\epsilon_{RG}
	:= u - u^\epsilon\,,
	\qquad\textrm{and}\qquad
	E^\epsilon_{Tr}
		:= u^\epsilon - v^\epsilon\,.
}{ErVa}

{\bf Step 1:} (\emph{RG error})
The equation for the error variable $e_{RG}^\epsilon$ reads,
\bsplitl{
\partial_t E^\epsilon_{RG}
	+ \ebrkts{\kappa\partial^2_x+\delta\partial_x^3+\nu\partial_x^4}E^\epsilon_{RG}
	+
	E^\epsilon_{RG}\partial_x u 
	+ 
	u^\epsilon\partial_x E^\epsilon_{RG}
	=\epsilon R_\epsilon(U)\,.
}{ErRG}
First, we test \reff{ErRG} with $-\partial_x^2 E^\epsilon_{RG}$, i.e.,
\bsplitl{
\partial_t \brkts{\partial_x E^\epsilon_{RG},\partial_x E^\epsilon_{RG}}
	& - \brkts{\kappa\partial^2_xE^\epsilon_{RG},\partial^2_xE^\epsilon_{RG}}
	- \brkts{\delta\partial_x^3E^\epsilon_{RG},\partial_x^2E^\epsilon_{RG}}
	-\brkts{\nu\partial_x^4E^\epsilon_{RG},\partial_x^2E^\epsilon_{RG}}
\\&
	-\brkts{
		E^\epsilon_{RG}\partial_x u,\partial_x^2 E^\epsilon_{RG}
	}	
	- \brkts{
		u^\epsilon\partial_x E^\epsilon_{RG},\partial_x^2 E^\epsilon_{RG}
	}
	= -\brkts{\epsilon R_\epsilon(U),\partial_x^2 E^\epsilon_{RG}}\,.
}{ErRG1}
Then, we use the test function $E^\epsilon_{RG}$,
\bsplitl{
\partial_t \brkts{E^\epsilon_{RG}, E^\epsilon_{RG}}
	& + \brkts{\kappa\partial^2_xE^\epsilon_{RG},E^\epsilon_{RG}}
	+\brkts{\delta\partial_x^3E^\epsilon_{RG},E^\epsilon_{RG}}
	+\brkts{\nu\partial_x^4E^\epsilon_{RG},E^\epsilon_{RG}}
\\&
	+\brkts{
		E^\epsilon_{RG}\partial_x u, E^\epsilon_{RG}
	}	
	+ \brkts{
		u^\epsilon\partial_x E^\epsilon_{RG},E^\epsilon_{RG}
	}
	= +\brkts{\epsilon R_\epsilon(U), E^\epsilon_{RG}}\,.
}{ErRG2}
Next, we add up \reff{ErRG1} and \reff{ErRG2} and apply the Sobolev embedding theorem
and standard inequalities to end up with,
\bsplitl{
\frac{1}{2}\frac{d }{d t}
	&\ebrkts{
		\Ll{E^\epsilon_{RG}}^2+\Ll{\partial_x E^\epsilon_{RG}}^2
	}
	+(\nu-3\alpha)\ebrkts{
		\Ll{\partial_x^2E^\epsilon_{RG}}^2
		+\Ll{\partial_x^3E^\epsilon_{RG}}^2
	}
\\&
	\leq C_{RG}(\kappa,\alpha,\epsilon,\N{u}{H^1},\N{\partial_x u}{H^1},\N{u^\epsilon}{H^1},\N{\partial_x u^\epsilon}{H^1})\N{E^\epsilon_{RG}}{H^1}^2\,.
}{ErRG3}
After defining
\bsplitl{
\tilde{C}_{RG}
	= 2C_{RG}\,,
}{tCrg}
we can multiply \reff{ErRG3} by ${\rm exp}\brkts{-\int_0^t\tilde{C}_{RG}\,ds}$ such that
\bsplitl{
\frac{1}{2}{\rm exp}\brkts{-\int_0^t\tilde{C}_{RG}\,ds}\frac{d}{dt}\Ll{E^\epsilon_{RG}}^2
	\leq {\rm exp}\brkts{
		-\int_0^t\tilde{C}_{RG}\,ds
	}\frac{\tilde{C}_{RG}}{2}\Ll{E^\epsilon_{RG}}^2\,.
}{ErRG4n}
Applying a corresponding identity based on the product rule as \reff{prRu}
in the proof of Lemma \ref{lem:ApEsRS}, we can simplify \reff{ErRG4n} to,
\bsplitl{
\frac{d}{dt}\brkts{
		{\rm exp}\brkts{-\int_0^t\tilde{C}_{RG}\,ds}\frac{1}{2}\Ll{E^\epsilon_{RG}}^2
	}
	\leq 0\,,
}{ErRG5n}
which further reduces by assumptions of Theorem \ref{thm:ErEs} and
after integration to,
\bsplitl{
\frac{1}{2}\Ll{E^\epsilon_{RG}(T)}^2
	\leq
	\frac{1}{2}{\rm exp}\brkts{-\int_0^T\tilde{C}_{RG}\,ds}\Ll{E^\epsilon_{RG}(0)}^2
	\leq C\Ll{E^\epsilon_{RG}(0)}^2\,.
}{ErRG6n}
In order to get a bound controlled by $\epsilon$ on the right-hand side of \reff{ErRG6n}, we have to
take into account the definition of the initial data \reff{InApSo1}, i.e.,
\bsplitl{
\Ll{E^\epsilon_{RG}(0)}^2
	= \epsilon^2\Ll{{\rm PF}_{NR}(0,g) }^2
	\leq C\epsilon^2\,,
}{ErRG5}
since $g\in H^2({\cal P}_\alpha)$ and hence ${\rm PF}_{NR}(0,g)\in H^1({\cal P}_\alpha)$
and its norm is bounded independently of $\epsilon=\frac{1}{N^4}$ in $H^1({\cal P}_\alpha)$.
Hence, we can conclude that,
\bsplitl{
\N{E^\epsilon_{RG}(t)}{L^2({\cg P}_\alpha)}^2
	\leq
	C\epsilon^2\,,
}{ErRG7}
which holds uniformly in time.

{\bf Step 2:} (\emph{Tr error}) We derive an estimate for the
error variable $E^\epsilon_{Tr}:= u^\epsilon-v^\epsilon= u^\epsilon-{\rm P}_Nu^\epsilon $. From \reff{PeKSeq} and \reff{PeKSeq2}, the error $E^\epsilon_{Tr}$ satisfies the equation,
\bsplitl{
\partial_t E^\epsilon_{Tr}
	+\ebrkts{\kappa\partial^2_x+\delta\partial_x^3+\nu\partial_x^4}E^\epsilon_{Tr}
	+{\rm P}_N[u^\epsilon u^\epsilon_x]
	-v^\epsilon v^\epsilon_x
	=\epsilon ({\rm P}_NR_\epsilon(U)-G_\epsilon(U,v^\epsilon),v^\epsilon)\,,
}{ErT}
which can be rewritten for all $\phi\in H^2_N$ by
\bsplitl{
\partial_t \brkts{E^\epsilon_{Tr},\phi}
	&+\brkts{\ebrkts{\kappa\partial^2_x+\delta\partial_x^3+\nu\partial_x^4}E^\epsilon_{Tr},\phi}
\\&
	+ \brkts{
		\brcs{{\rm P}_N[u^\epsilon u^\epsilon_x] - v^\epsilon v^\epsilon_x},E^\epsilon_{Tr}
	}
	=\epsilon \brkts{
		\brcs{{\rm P}_NR_\epsilon(U)-G_\epsilon(U,v^\epsilon)},\phi
	}\,.
}{ErT1}
Choosing $\phi=E^\epsilon_{Tr}$ allows to estimate \reff{ErT1} in the following way,
\bsplitl{
\frac{1}{2}\frac{d}{dt}\Ll{E^\epsilon_{Tr}}^2
	+\nu \Ll{\partial_x^2 E^\epsilon_{Tr}}^2
	& \leq
	C(\alpha,\kappa)\Ll{E^\epsilon_{Tr}}^2 +\alpha\Ll{\partial_x E^\epsilon_{Tr}}^2
	+ (I)
	+(II)\,,
}{ErT2}
where we define,
\bsplitl{
(I)
	& := \brkts{
		{\rm P}_N[u^\epsilon u^\epsilon_x] - vv_x,E^\epsilon_{Tr}
	}\,,
\\
(II)
	& := \epsilon\brkts{
		\Ll{{\rm P}_NR_\epsilon(U)} + \Ll{G_\epsilon(U,v^\epsilon)}
	}\Ll{E^\epsilon_{Tr}}\,.
}{1and2}
Let us first control term $(I)$, that means,
\bsplitl{
\av{(I)}
	& \leq
	\av{
		\brkts{
				-{\rm P}_Nu^\epsilon({\rm P}_Nu^\epsilon)_x+u^\epsilon u^\epsilon_x,E^\epsilon_{Tr}
			}
	}
	+
	\av{
		\brkts{
				{\rm P}_Nu^\epsilon({\rm P}_Nu^\epsilon)_x-v^\epsilon v^\epsilon_x,E^\epsilon_{Tr}
			}
	}
\\&
	\leq
	\frac{1}{2}\av{
		\brkts{
			\partial_x\brkts{(u^\epsilon-{\rm P}_Nu^\epsilon)(u^\epsilon+{\rm P}_Nu^\epsilon)},E^\epsilon_{Tr}
		}
	}
	+ 1/4\av{
		\brkts{
			\partial_x\brkts{{\rm P}_Nu^\epsilon+v^\epsilon},E^\epsilon_{Tr}
		}
	}
\\&
	\leq
	C\brkts{
		\N{u^\epsilon}{H^1}
		+ \N{{\rm P}_Nu^\epsilon}{H^1}
	}
	\N{u^\epsilon-{\rm P}_Nu^\epsilon}{H^1}
	\Ll{E^\epsilon_{Tr}}
	+C\Ll{E^\epsilon_{Tr}}^2\,,
}{I}
where we used the embedding $H_N^1({\cg P}_\alpha)$ into $L^\infty({\cg P}_\alpha)$,
i.e., $\N{\partial_x({\rm P}_Nu^\epsilon +v^\epsilon)}{L^\infty}\leq C$. The second
term $(II)$ immediately becomes,
\bsplitl{
\av{(II)}
	\leq \epsilon{\rm exp}\brkts{-ct/\epsilon}C\Ll{E^\epsilon_{Tr}}\,.
}{II}
Define
\bsplitl{
\gamma(t)
	& := C(\Ll{u^\epsilon},\Ll{v^\epsilon},\N{u^\epsilon}{H^1},\N{{\rm P}_Nu^\epsilon}{H^1},
			\N{u^\epsilon}{H^2})\epsilon^{1/4}
\\&\quad
	+C(\Ll{u^\epsilon},\Ll{v^\epsilon})\epsilon{\rm exp}\brkts{-Ct/\epsilon}\,,
\\
\beta
	& := 2\brkts{C(\alpha,\kappa)+C}\,,
}{ga}
and mulitply \reff{ErT2} with ${\rm exp}\brkts{-\int_0^t\beta\,ds}$ such that
\bsplitl{
\frac{1}{2}{\rm exp}\brkts{
		-\beta t
	}\frac{d}{dt}\Ll{E^\epsilon_{Tr}}^2
	\leq \frac{\beta}{2}{\rm exp}\brkts{-\beta t}\Ll{E^\epsilon_{Tr}}^2
	+\gamma(t){\rm exp}\brkts{
		-\beta t
	}\,.
}{ErT3n}
Using again a corresponding identity to \reff{prRu} in Lemma \ref{lem:ApEsRS} we can
rewrite \reff{ErT3n} as,
\bsplitl{
\frac{d}{dt}\brkts{
		{\rm exp}\brkts{
			-\beta t
		}\frac{1}{2}\Ll{E^\epsilon_{Tr}}^2
	}
	\leq \gamma(t){\rm exp}\brkts{
		-\beta t
	}\,,
}{ErT4}
which becomes after integration with respect to time,
\bsplitl{
\frac{1}{2}{\rm exp}\brkts{
		-\beta T
	}\Ll{E^\epsilon_{Tr}}^2(T)
	& \leq C\int_0^T\brcs{
		\brkts{
			\epsilon^{1/4}
			+\epsilon{\rm exp}\brkts{-Ct/\epsilon}
		}{\rm exp}\brkts{-\beta t}
	}\,dt
\\
	& \leq
	\frac{C\epsilon^{1/4}}{\beta}\brkts{
		1-{\rm exp}\brkts{\beta T}
	}
	+\frac{C\epsilon}{\beta+C/\epsilon}\brkts{
		1-{\rm exp}\brkts{-(\beta+C/\epsilon)T}
	}\,.
}{ErT5}
In the remaining part we want to improve \reff{ErT5} with the help of
Gevrey spaces. To this end, we remark that the factor $\epsilon^{1/4}$
in \reff{ga} relies on the interpolation estimate \reff{IntEst}. If we assume
that solutions $u^\epsilon$ are in $G_{\sigma, s}$, then we improve
\reff{IntEst} by \reff{Gint}. As a consequence, we are able to rewrite inequality
\reff{ErT5} by
\bsplitl{
\frac{1}{2}\Ll{E^\epsilon_{Tr}}^2(T)
	& \leq C\int_0^T
		\Bigl(
			\epsilon^{1/4}{\rm exp}\brkts{\beta (T-t)-\frac{\sigma}{\epsilon^{1/4}}}
			+\epsilon{\rm exp}\brkts{-Ct/\epsilon+\beta(T-t)}
		\Bigr)
	\,dt
\\
	& \leq \frac{C\epsilon^{1/4}}{\beta}{\rm exp}\brkts{\beta T}
		+\frac{C\epsilon}{\beta+c/\epsilon}{\rm exp}\brkts{\beta T}
	\,.
}{ErT6}
\end{proof}

%

\section{Stochastic mode reduction}\label{sec:StMoRe}
In this section the renormalized equations \reff{rgeq1}, \reff{Wj}
from Sections \ref{sec:prEqs} - \ref{sec:Cmp} allow for a rigorous
stochastic mode reduction similar in spirit to the Mori-Zwanzig one
\cite{Mori1965,Zwanzig1961} but for systems not satisfying an
extended or generalized Hamiltonian structure \cite{Zwanzig1973} and
without \new{a} canonical invariant measure.

As in the \new{Mori-Zwanzig} formalism, we assign a stochastic process to
the unresolved modes. This is done by applying Jaynes maximum
entropy principle, see
\cite{Jaynes1957,Jaynes1957a,Rosenkrantz1989}. This seems a
reasonable approach for our problem since we do not have a
canonically induced probability density. Hence maximizing the
information entropy for the probability density of Fourier modes is
equivalent to maximizing the multiplicity of Fourier modes.
Multiplicity means the number of different ways a certain state in a
system can be achieved. States in a system with the highest
multiplicity can be realized by nature in the largest number of
ways. Hence, the probability density
functions with maximum entropy are optimal statistical descriptions. 

It should also be noted that a system at equilibrium will most
probably be found in the state of highest multiplicity since
fluctuations from that state will be usually too small to measure.
The probability distribution may also be obtained from experiments
as statistical data. In \cite{Stinis}, the probability distribution
is constructed by a conditional expectation obtained from previously
computed samples which are used to fit an a priori assumed Gaussian
distribution.

Finally, we emphasize that the maximum entropy principle can also be
applied to problems where one lacks deterministic data as a
consequence of not enough experimental data to fix all degrees of
freedom. A common approach to model such uncertainty is to use white
noise. The maximum entropy method turns out to be an attractive
alternative because it allows to systematically add noise to 
the gKS equation over the equation \reff{PeKSeq2} which is obtained 
by the evolutionary RG method.

However, since we apply the entropy maximization principle
\cite{Jaynes1957} on an approximate equation, we already neglect
information from the beginning and hence have to account for this by
an asymptotic in time characterization of the fast modes for
example, see Assumption (A) below. This assumption might be improved
or adapted appropriately in other applications. \new{Subsequently, 
$(\Omega,{\cal F},{\mathbb P})$ denotes the usual probability space with sample space 
$\Omega$, $\sigma$-algebra ${\cal F}$, and probability measure $\mathbb{P}$.}

\subsection{Problem induced probability density by maximizing information entropy}\label{sec:MaEn}
With the considerations at the beginning of Section
\ref{sec:StMoRe}, we assign a probability distribution to the
unresolved degrees of freedom $W$ based on the following

\medskip

{\bf Assumptions:}\\
\begin{itemize}
\item[{\bf (A)}] For a \new{probability measure ${\mathbb P}_j$ with density $f_j$ and}
\bsplitl{
\mathcal{C}_N(\tilde{w}^{\epsilon,0}_j)
	:= \frac{1}{2}\brkts{\tilde{w}^{\epsilon,0}_j}^2\,,
}{ToEn}
\new{where $\tilde{w}^{\epsilon,0}_j(t)=w^{\epsilon,0}_j(t,\omega)$ denotes a realization for $\omega\in\Omega$ of the $j$-th Fourier mode of the 
leading order term $w^{\epsilon,0}$ of $w^\epsilon$ in \reff{ApSo1}$_2$, i.e., 
$w^{\epsilon,0}(x,t)={\rm e}^{-{\rm Q}_Nt/\epsilon}W(x,t)$,} 
we assume that it holds asymptotically in time that
\bsplitl{
\mathbb{E}_j\ebrkts{\frac{\partial}{\partial t}\mathcal{C}_N(\new{w^{\epsilon,0}_j})}
	= \int_{-\infty}^{\infty}f_j(\tilde{w}^{\epsilon,0}_j)\frac{\partial}{\partial t}\mathcal{C}_N(\tilde{w}^{\epsilon,0}_j)\,d\tilde{w}^{\epsilon,0}_j
	= \delta_j(t)
	:=-\rho^w_j{\rm e}^{-2\rho^w_jt}W_j^2(0)
	\,,
}{ToEn2}
i.e., there is a $t_0\geq 0$ such that \reff{ToEn2} holds for $t>t_0$. \new{We call $\delta_j(t)$ 
a dissipation rate and $\mathbb{E}_j$ denotes the expectation with respect to the probability 
$\mathbb{P}_j$.}
\item[{\bf (B)}] Under {\bf (A)} the \new{probability ${\mathbb P}_j[W_{j}\leq \tilde{w}^{\epsilon,0}_j]
	=F(\tilde{w}^{\epsilon,0}_j)$ with density $f_j$}, i.e.,
\bsplitl{
F(\tilde{w}^{\epsilon,0}_j)
	:=\int_{-\infty}^{\tilde{w}^{\epsilon,0}_j}f_j(r)\,dr\,,
}{P}
has maximum
\emph{information entropy} $\mathbb{S}_I(f_j)$,
\bsplitl{
\mathbb{S}_I(f_j)
	= - \int_{-\infty}^\infty f_j(r){\rm log}\brkts{
		\frac{f_j(r)}{\nu(r)}
	}\,dr\,,
}{SI}
where $f_j(r)$ denotes the probability density of the $j$-th Fourier mode of the 
unresolved variable $W$ and $\nu$ is an according \emph{invariant
measure} which is defined on background information intrinsically
given by the physical origin of $W$.
\end{itemize}

\medskip

\begin{rem}\label{rem:EqPaEn}
1) The idea of deriving probability distributions for multiscale
evolution problems by maximizing the information entropy seems to go
back to \cite{Lorenz1996}. The energy argument in \cite{Lorenz1996},
which assumes that the fast modes reached already the stationary
state, does not provide here enough information to fix the Lagrange
multiplier $\lambda_1$ associated with this energy constraint. We
impose \emph{Assumption {\bf (A)} } instead. \new{Note that we take 
slightly  more information into account by using $w^{\epsilon,0}$ instead 
of $W$ which does not decay as fast as $w^{\epsilon,0}$.}
\\
2) A mechanical system governed by the Hamiltonian $H(q,p)$
canonically induces an invariant measure by the density
distribution function $f(q,p):=\frac{1}{Z(\beta)}{\rm e}^{-\beta H(q,p)}$.\\
3) Equation \reff{ToEn2} accounts for the fact that we do not have an invariant 
measure to the fast modes. For simplicity, we also neglect a possible randomness 
in time. This is a further reason for the assumption in
\reff{ToEn2}.
\end{rem}

In information theory, an entropy related to \reff{SI} was
originally introduced by Shannon \cite{Shannon1948} to measure the
maximum information content in a message. The Assumptions {\bf (A)}
and {\bf (B)} above account for the lack of a free energy and a
Hamiltonian for which the thermodynamic equilibrium (invariant
measure) can be achieved via the gradient flow with respect to the
Wasserstein distance \cite{Jordan1998}. In fact, it should be noted
that minimizing the free energy with respect to constant internal
energy is equivalent to maximizing the entropy.

\medskip

\begin{theorem}\label{thm:PrDe}
Under Assumptions {\bf (A)} and {\bf (B)}, it follows that the unresolved modes
$W_k$ for $\av{k}>N$ obtained by equation \reff{Wj} are normally
distributed with zero mean, i.e., $\mu_k=0$, and variance
$\sigma_k^2=\frac{1}{2\lambda_k\rho_k}$, where
$\lambda_k:=\frac{1}{2\delta_j(t)}$ 
is a Langrange
multiplier.
\end{theorem}

\medskip

\begin{rem}
Instead of {\bf (A)}, one can make the following assumption {\bf (A*)}: For large enough times $t>0$, it holds that
\bsplitl{
\mathbb{E}_j\ebrkts{\frac{\partial}{\partial t}\mathcal{C}_N(\new{w^{\epsilon,0}_j})}
	= \sigma^2\,.
}{Aaltern} 
This immediately leads to the result 
that the fast modes satisfy $W_k\sim{\cal N}(0,\sigma^2)$ where the variance can be 
defined by the power spectral density as in the case of complete uncertainty, see also 
Section \ref{sec:DiAp}.
\end{rem}

\medskip
To keep the considerations simple, we only account for a spatial
random process and keep the time deterministic in Theorem
\ref{thm:PrDe} (and Assumption {\bf (A)}).

\medskip

\begin{proof}
To maximize \reff{SI} under Assumptions {\bf (A)} and
{\bf (B)}, we apply the following constraints:
\bsplitl{
{\bf (CI)}\qquad
\begin{cases}
\int_{-\infty}^\infty f_k(r)\,dr
	= 1\,,
	&
\\
\mathbb{E}_k[\partial_t\mathcal{C}_N(w^{\epsilon,0}_k)]
	:= \int f_k(\tilde{w}^{\epsilon,0}_k)
		\partial_t\mathcal{C}_N(\tilde{w}^{\epsilon,0}_k)\,d\tilde{w}^{\epsilon,0}_k
	= \delta_k(t)
	\,,
\end{cases}
}{CI} where \reff{CI}$_2$ is a consequence of assumption (i).

Hence, maximizing the entropy $\mathbb{S}_I$
subject to the constraints \reff{CI} leads to
\bsplitl{
\int_{-\infty}^\infty \delta f_k(\tilde{w}^{\epsilon,0}_k)\brcs{{\rm log}\brkts{\frac{f_k(\tilde{w}^{\epsilon,0}_k)}{\nu(\tilde{w}^{\epsilon,0}_k)}}
		+\nu(\tilde{w}^{\epsilon,0}_k)+\lambda_0
		+\lambda_k \partial_t\mathcal{C}_N(\tilde{w}^{\epsilon,0}_k)}
	\,d\tilde{w}^{\epsilon,0}_k = 0\,,
}{MxSI}
where $\lambda_0$ and $\lambda_k$ are Lagrange multipliers
associated with the constraints \reff{CI}.
In order to give \reff{CI}$_2$ a precise meaning, we write down the explicit form
of the equation belonging to each Fourier coefficient of the fast
 mode variable $w={\rm Q}_Nu$ solving \reff{rgeq1}.

We briefly show what the
constraint \reff{CI}$_2$ means,
\bsplitl{
\mathbb{E}_k[\partial_t\mathcal{C}_N(w^{\epsilon,0}_k)]
	& = 
		-\int_{-\infty}^\infty f_k(\tilde{w}^{\epsilon,0}_k)
		\brkts{\tilde{w}^{\epsilon,0}_k}^2
		\brkts{
			2i\lambda\frac{k}{\alpha}V_0
			-\frac{1}{\epsilon}\rho^w_k
		}\,d\tilde{w}^{\epsilon,0}_k
\\&
	= \delta_k(t)
	\,.
}{ptE}
We recall that $V_0={\rm const.}$ due to conservation of mass. 
Since \reff{MxSI} should hold
for arbitrary variations $\delta f_k$, we obtain the following expression
for the probability density function,
\bsplitl{
f_k(\tilde{w}^{\epsilon,0}_k)
	= \frac{1}{Z_k}\nu(\tilde{w}^{\epsilon,0}_k) {\rm e}^{
		-\lambda_k\partial_t\mathcal{C}_N(\tilde{w}^{\epsilon,0}_k)}\,, 
}{f} 
where $Z_k:={\rm e}^{\nu(\tilde{w}^{\epsilon,0}_k)+\lambda_0}$ is called the ``partition
function" which is determined by the normalization constraint
\reff{CI}$_2$, i.e., 
\bsplitl{ 
Z_k
	:= Z_k(\lambda_k)
	= \int_{-\infty}^\infty\nu_k(\tilde{w}^{\epsilon,0}_k){\rm e}^{-
		\lambda_k \partial_t\mathcal{C}_N(\tilde{w}^{\epsilon,0}_k)}\,d \tilde{w}^{\epsilon,0}_k\,.
}{NC}
Since the constraint \reff{CI}$_2$ is quadratic in its nature,
we represent it by
\bsplitl{
\lambda_k\partial_t\mathcal{C}_N(\tilde{w}^{\epsilon,0}_k)
	= -\lambda_k
	\brkts{\tilde{w}^{\epsilon,0}_k}^2
	\tilde{\rho}^w_k
	= 
	-\frac{1}{2\sigma_k^2}\brkts{
		(\tilde{w}^{\epsilon,0}_k-\mu_k)^2
		-\mu_k^2
	}
	\,,
}{CRpr}
where $\tilde{\rho}^w_k:=2i\lambda\frac{k}{\alpha}V_0-\rho^w_k/\epsilon$ and 
\bsplitl{
\sigma^2_k
	= \frac{1}{2\lambda_k\tilde{\rho}_k^w}\,,
\qquad\textrm{and}\qquad
\mu_k
	= 
	0
	\,.
}{abc}
With identities \reff{CRpr} and \reff{abc} the probability density
function \reff{f} can be written as
\bsplitl{
f_k(\tilde{w}^{\epsilon,0}_k)
	= \frac{1}{Z_k}c^{-1}_{W_k}
			\sigma_k\sqrt{2\pi}
		\mathcal{N}(\mu_k,\sigma_k,\tilde{w}^{\epsilon,0}_k)\,,
}{f(s)}
for $\av{k}>N$ where $\mathcal{N}$ is the normal distribution given by
\bsplitl{
\mathcal{N}(\mu_k,\sigma_k,\tilde{w}^{\epsilon,0}_k)
	= \frac{1}{\sigma_k\sqrt{2\pi}}
	{\rm e}^{-\frac{(\tilde{w}^{\epsilon,0}_k-\mu_k)^2}{2\sigma^2_k}}\,,
}{N}
which is characterized by the following moments
\bsplitl{
\int_{-\infty}^\infty\mathcal{N}(\mu,\sigma,w)\,dw = 1\,,
\quad&
\int_{-\infty}^\infty\mathcal{N}(\mu,\sigma,w)w\,dw = \mu\,,
\quad\textrm{and }
\\&
\int_{-\infty}^\infty\mathcal{N}(\mu,\sigma,w)w^2\,dw = \sigma^2+\mu^2\,.
}{prN}
The first property in \reff{prN} together with the normalization condition
\reff{CI}$_1$ allow us to define the partition function $Z$ by
\bsplitl{
Z_k = c^{-1}_{W_k}
                        \sigma_k\sqrt{2\pi}
	\,,
}{Z}
for $\av{k}>N$.

\medskip

\begin{rem}
The measure $\nu({\bf w}):= \prod_{\av{k}>N}\frac{1}{c_{W_k}}$ is the probability density function if we only have
a priori information. Usually, it is a nontrivial task and basic considerations
of symmetries are required to find this measure $\nu$. 
\end{rem}

\medskip

The probability density function $f_k$ admits then the simple
form as a product of Gaussian distributions, i.e., 
\bsplitl{ 
f_k(\tilde{w}^{\epsilon,0}_k) 
	= \mathcal{N}(\mu_k,\sigma_k,\tilde{w}^{\epsilon,0}_k)\,, 
}{G} 
for $\av{k}>N$.
With the second and third property in \reff{prN} and the constraint
\reff{CI}$_2$, i.e., \reff{ptE}, we obtain for all $\av{k}>N$,
\bsplitl{
\delta_k(t)
	& 
	= \mathbb{E}_k[\partial_t\mathcal{C}_N(\omega_k)]
	= -\brkts{
		\tilde{\rho}_k^w\brcs{\sigma_k^2+\mu_k^2}
	}
	\,.
}{CIIeq}
We conclude with \reff{CIIeq} that the Lagrange parameter $\lambda_k$ is,
\bsplitl{
\lambda_k(t)
	=
	{\rm e}^{2\rho_k^wt}/(2\rho_k^wW^2_0(0))
	\,.
}{lambda}
The important information contained in formula \reff{lambda} and \reff{abc}$_1$ is that we do not have to assert to
each Fourier mode $k$ its standard deviation $k^2\sigma_W^2$. We only need to determine once the
Lagrange paramter $\lambda_1$ via \reff{lambda}. From \reff{CIIeq} and the property $\overline{\mu_k}=\mu_{-k}$, we obtain that the mean satisfies 
$\mu_k=0$ for all $\av{k}>N$.
\end{proof}

Hence, the approach of maximizing the generalized information
entropy allows to systematically determine the probability
distribution function $f_k(w_k)$ of the Fourier modes for the
unresolved degrees of freedom $W$. The stochastic partial differential 
equation for the resolved degrees of freedom is then obtained by 
computing the probability distribution of $W$ as the inverse Fourier transform of 
the sum of normally distributed unresolved Fourier modes and by assuming that 
the probability distribution for $W$ derived in the long time regime also holds 
for the unresolved modes of the initial conditions.
\medskip

\medskip

We emphasize that the RG approach suggests multiplicative noise as a
compensation for the unresolved modes unlike the commonly obtained
additive noise by Mori-Zwanzig's mode reduction \cite{Zwanzig1973}.
Moreover, an estimate \reff{ApSolEs}, which can be correspondingly
derived by additionally accounting for the Galerkin error, shows
that the influence of the stochastic force decreases for decreasing
$\epsilon:=\frac{1}{N^4}$.

\section{Direct approach: Replacement of $G_\epsilon$ by white noise}\label{sec:DiAp}
The result of Lemma \ref{lem:Gest} also enables for a direct
approach to model the unresolved degrees of freedom as completely
unknown. Such kind of complete uncertainty is generally described by
white noise $W(x)$ with zero mean and a variance equal to the power
spectral density. It is very common and widely accepted to model
uncertainty by white noise. Hence, we replace $\epsilon
G_\epsilon(U(t),v^\epsilon)$ in equation \reff{PeKSeq2} by \bsplitl{
{\cg N}_\epsilon(x,t) := \epsilon {\rm exp}\brkts{-Ct/\epsilon}
W(x)\,, }{Wnoise} where $W(x)\in L^2({\cg P}_\alpha)$ is the
Gaussian random variable as motivated above, i.e., with zero mean
$\mu$ and suitable variance $\sigma$. It is immediately clear that
${\cg N}_\epsilon$ is a compatible replacement of $G_\epsilon$ since
\reff{Wnoise} satisfies a bound corresponding to the one in Lemma
\ref{lem:Gest}. One can follow Stinis' approach \cite{Stinis} for
example in order to determine $\mu$ and $\sigma$ by a maximum
likelihood method.

\section{Discussion and conclusions}\label{sec:Cocl}

We have formally developed a new stochastic mode reduction strategy
with a rigorous basis by obtaining appropriate error estimates. The
analysis can be summarized in three key steps as follows:

\emph{ (1) RG method:} The RG technique \cite{Temam1999,Moise2001}
turns out to be a formal and feasible method to decompose the gKS
equation into slow $v^\epsilon$ and fast variables $w^\epsilon$,
respectively. The equation for the slow modes $v^\epsilon$
represents a Galerkin approximation of the gKS equation plus an
additional perturbed force term $\epsilon G_\epsilon$ which also
depends on the infinite dimensional renormalized fast modes $W$.  An
important property of the RG technique is that it can be easily
extended to higher space dimensions, see \cite{Temam1999} with respect
to the RG method and \cite{Biswas2007} for an existence theory of the
KS equation in higher space dimensions. We also
remark that the dispersion term, i.e. $u_{xxx}$, does not affect the
mode reduction analysis.

\emph{(2) Error bounds:} We rigorously characterize the formal RG
method (1) by qualitative error estimates (Theorem \ref{thm:ErEs}).
These estimates further allow for an additional direct mode
reduction strategy which is much simpler and straightforward but not
as systematic. The basic idea is to replace the perturbed force term
$\epsilon G_\epsilon$ directly by white noise. A physical motivation
for such a simplified reduction is the fact that white noise is a
well-accepted random model for complete uncertainty.

\emph{(3) Maximum entropy principle:} Due to the lack of a
Hamiltonian structure and an invariant measure, we apply Jaynes'
maximum entropy principle \cite{Jaynes1957,Jaynes1957a} to define
the renormalized fast modes $W$ as a random variable. This random
variable then, together with the renormalized approximation of the
slow variable $v^\epsilon$, provides a systematic explanation for
the appearance of a noisy low dimensional gKS equation. In contrast
to optimal prediction we obtain optimality in the sense of maximum
entropy here.

\medskip

\medskip

There are three main features of the new low dimensional gKS equations:

\emph{ (i)~Reliable and efficient numerics:} The low dimensional
formulation developed here should allow for reliable (since
information from the unresolved degrees of freedom included) and
efficient (since low dimensional) numerical approximations. In fact,
we systematically account for the unresolved degrees of freedom by
the steps (1) and (2) above. This is especially of importance since
the choice of slow and fast variables depends on the physical
problem and is often not clear. For instance, by considering the gKS
in large domains, it is possible to introduce a further scale which
accounts for the unstable modes. Hence, we can study three different
scales such as ``unstable modes'', ``slow stable modes'', and ``fast
stable modes''. The main question is then how to account for the
unstable and the fast stable modes in an equation for the resolved
slow modes only.

Moreover, the error estimates from step (2) provide a qualitative
measure on how to choose the dimension of the slow variable. This is
also the main advantage of mode reduction considerations over pure
convergence analyses of Galerkin approximations (e.g. numerical
schemes) where one completely neglects the unresolved degrees of
freedom. Hence, straightforward discretization strategies might
lose model relevant information in the neglected degrees of
freedom. This is a major motivation to include rigorous mode
reduction strategies as an important part of the development of
computational schemes. We further remark that this is a major reason
why the addition of noise to deterministic partial differential
equations  shows good results and is currently a topic of increasing
interest. \new{It is also important to emphasize that one of the key
points for the presented methodoly to be compuationally efficient is precisely
because we add the noise {\it a posteriori} after solving the reduced model, something
which is computationally simpler than solving the full system at every time step.}

\emph{ (ii)~No Hamiltonian structure; no invariant measure:} Many classical mode reduction strategies
rely either on a Hamiltonian structure or an invariant measure. Based on the three steps (1)-(3) above,
the new asymptotic reduction strategy circumvents such dependences. For example, when classical
optimal prediction methods \cite{Chorin1998} fail because of such deficiencies, the stochastic renormalization
provides optimality in the sense of maximum information entropy and hence proves as a promising
alternative.

\emph{ (iii)~The role of noise:} We gain a rigorous understanding of the origin of noise and the way it
appears in the gKS equation. This is especially of interest due to numerical
evidence provided together with a heuristic motivation in \cite{CHOW1995} for instance.

\medskip

Clearly, there are open questions and future perspectives. For
example, motivated by the comparative study initiated by Stinis
\cite{STINIS2006}, it would be of interest to numerically analyze
and compare available mode reduction strategies such as adiabatic
elimination \cite{VanKampen1985}, invariant manifolds
\cite{Foias1988}, and optimal prediction \cite{Chorin1998} with the
new RG approach developed here. Since the statistically based
optimal prediction \cite{Stinis} is performed for a truncated KS
equation, it provides a convenient setup for comparison with the new
and more generally applicable method suggested here.

Another question is how can we apply the RG method to the derivation of a
low-dimensional approximation for a gKS equation investigated under three
scales, i.e., ``slow unstable modes'', ``slow stable modes'', and ``fast
stable modes'' or to explore the possibility of obtaining low-dimensional
approximations of equations where noise is present from the outset,
e.g.~\cite{Marc_EJAM,Marc_PRL}

The RG method is based on a natural splitting into linear and
nonlinear terms by the variation of constants formula. Recent
studies, e.g. by Holden et al. \cite{Holden2011a,Holden2011}, make
use of such a splitting via a suitable numerical scheme for
equations with Burgers' nonlinearity. Hence, the reliability and
efficiency of the renormalized low dimensional gKS equation motivate
the application such numerical splitting strategies to the new
reduced equations derived here. Finally, we emphasize that efficient
low dimensional approximations are of great interest for numerical
scrutiny of long time asymptotes. We shall examine these and related
issues in future studies.


\section*{Acknowledgements}
We thank Dirk Bl\"omker (U Augsburg) for drawing our attention to reference \cite{Bloemker2002}. We acknowledge financial support from EPSRC Grant No. EP/H034587,
EU-FP7 ITN Multiflow and ERC Advanced Grant No. 247031. 



\medskip
\begin{center}
{\bf Appendix:}
\end{center}

We prove the following \\
\medskip

{\bf Lemma:} $R_2(j)$ is an empty set.
\\
\medskip
\begin{proof}
Let us first recall the definition of the second resonance, that is,
\begin{equation}
\rho_k^w + \rho_l^w = \rho_j^w,
\tag{A.1}
\label{A1}
\end{equation}
where the indices satisfying \reff{A1} belong to $R_2(j)$. We further remind the convention 
according to (\ref{EV})
 \begin{equation}
 \rho_k^w = \frac 1{N^4} \brkts{ -\nu \brkts{\frac k\alpha}^2 - i\delta \brkts{\frac
k\alpha}^3 + \kappa \brkts{\frac k\alpha}^4}
\tag{A.2}
\label{A2}.
\end{equation}
Taking the imaginary parts of \reff{A1}, one obtains
\begin{equation}
\brkts{\frac k\alpha}^n + \brkts{\frac l\alpha}^n = \brkts{\frac j\alpha}^n,
\tag{A.3}
\label{A3}
\end{equation}
for $n=3$. Taking the real parts, one only gets
\begin{equation} -\nu \brkts{\brkts{\frac k\alpha}^2 + \brkts{\frac l\alpha}^2} 
	+ \kappa \brkts{\brkts{\frac
k\alpha}^4 + \brkts{\frac l\alpha}^4}
= -\nu \brkts{\frac j\alpha}^2 + \kappa \brkts{\frac j\alpha}^4.
\tag{A.4}
\label{A4}
\end{equation}

In what follows, we
transform (\ref{A4}) into an expression which convinces us that there are no indices
$k$ and $l$ that satisfy \reff{A4}. To this end, we make use of the fact that
$j=k+l$ which reads after taking the square on each side as
\begin{equation}
j^2=(k+l)^2=k^2+l^2+2kl\,.
\tag{A.5}
\label{j2}
\end{equation}
Multiplying now \reff{A4} by
$\alpha^2$ gives
\begin{equation}
-k^2-l^2
	+\frac{\kappa}{\alpha^2\nu}(k^4+l^4)
	=-j^2
	+\frac{\kappa}{\alpha^2\nu}j^4
\,,
\tag{A.6}
\end{equation}
and after applying (\ref{j2}) on the right-hand side we obtain,
\begin{equation}
-k^2-l^2
	+\frac{\kappa}{\alpha^2\nu}(k^4+l^4)
	=-k^2-l^2 -2kl
	+\frac{\kappa}{\alpha^2\nu}\Bigl(
		k^2+l^2+2kl
	\Bigr)^2
\,,
\tag{A.7}
\end{equation}
and hence becomes
\begin{equation}
r-3kl
	=
	2(k^2+l^2)
\,,
\tag{A.8}
\label{star}
\end{equation}
where we set $r:=\frac{\alpha^2\nu}{\kappa}$ which is positive.

Equation (\ref{star}) cannot be satisfied by any integers $k$ and $l$ 
if (i) $N^2\geq \frac{8}{7}r$ or 
if (ii) $L$, $\nu$, and $\kappa$ are not proportional  to $\pi$.
\hfill
\end{proof}

\medskip

{\bf Remark A.1.} \emph{
We note that without either assuming that (i) $L$, $\nu$, $\kappa$ are not proportional to $\pi$ or that (ii) $N^2\geq \frac{8}{7}\frac{\alpha^2\nu}{\kappa}$, we obtain two explicit solutions for $k$ and $l$ via (\ref{star}) 
over a depressed cubic equation, i.e.,
\bsplit{
k=\brkts{r(-64\mp\sqrt{64^2+392^2r/27})}^{1/3}
	- \frac{392r}{3\brkts{r(-64\mp\sqrt{64^2+392^2r/27})}^{1/3}}\,,
}
where the same expression also defines $l$. Herewith, it leaves to check whether for a given $r\in\mathbb{R}$ the 
solutions $k$ and $l$ are integers and whether they satisfy $|k+l|>N$, $|k|>N$ 
and $|l|>N$.
}

\end{document}